\documentclass[12pt]{article}
\usepackage{amsmath,amssymb,amsthm,mathtools,geometry}
\usepackage{graphicx,enumitem,color,threeparttable,booktabs,multirow,makecell}
\usepackage[title]{appendix}
\usepackage[hyphens]{url}
\usepackage{hyperref}
\newtheorem{theorem}{Theorem}

\newtheorem{lemma}{Lemma}
\newtheorem{assumption}{Assumption}

\newtheorem{algorithm}{Algorithm}

\allowdisplaybreaks[4]

\newcommand{\E}{\mathrm{E}}
\newcommand{\rP}{\mathrm{P}}
\newcommand{\Var}{\mathrm{Var}}
\newcommand{\Cov}{\mathrm{Cov}}

\newcommand{\rd}{\mathrm{d}}
\newcommand{\pto}{\stackrel{p}{\longrightarrow}}
\newcommand{\dto}{\stackrel{d}{\longrightarrow}}
\newcommand{\wc}{\mathcal{W}_{k,k^\prime}^c}

\usepackage[round]{natbib}
\newcommand{\anon}{1}

\begin{document}

\def\spacingset#1{\renewcommand{\baselinestretch}%
{#1}\small\normalsize} \spacingset{1}


\if1\anon
{
  \title{\bf Double/Debiased Machine Learning for Continuous Treatment Effects in Panel Data with Endogeneity}
  \author{Peikai Wu\\
    Department of Statistics and Data Science, Fudan University, China\\
    Kuan Sun\\
    Department of Statistics and Data Science, Fudan University, China\\
    and \\
    Zhiguo Xiao \\
    Department of Statistics and Data Science, Fudan University, China}
  \maketitle
} \fi

\if0\anon
{
  \bigskip
  \bigskip
  \bigskip
  \begin{center}
    {\LARGE\bf Double/Debiased Machine Learning for Panel Data in the Presence of Endogeneity}
\end{center}
  \medskip
} \fi

\bigskip
\begin{abstract}
We propose a double/debiased machine learning framework to estimate average derivative effects in nonparametric panel models with two-way fixed effects. It extends instrumental variable methods to panel settings, handles continuous treatments and various forms of endogeneity, and introduces a cross-fitting scheme to restore independence after eliminating time fixed effects. A penalized GMM debiasing term enables automatic debiased machine learning with endogeneity. Our estimators for contemporaneous, dynamic, and aggregated effects are consistent and asymptotically normal with a valid variance estimator. Simulations show reduced regularization bias and accurate confidence intervals. An application to ECLS-K data reveals rich dynamics in the effect of family SES on childhood BMI. 
\end{abstract}

\noindent%
\vfill

\newpage
\spacingset{1.8} 

\section{Introduction}
\label{sec-intro}
Researchers are frequently interested in the effect of a continuous, time-varying treatment in panel data. As linear model  specifications risk misspecification, semiparametric and nonparametric panel models (see \citet{rodriguez2017nonparametric,su2011nonparametric} for reviews) have garnered increasing attention. In these models, the treatment effect is defined as a functional of the structural function that determines how the treatment, covariates, and other variables generate the outcome. The problem is typically high-dimensional, and endogeneity is ubiquitous: regressors may be correlated with the contemporaneous error, uncorrelated with the contemporaneous error but correlated with future errors, or include lagged outcomes. The combination of model flexibility, high-dimensionality, and endogeneity calls for estimation tools that support both consistent estimation and valid inference for the treatment effect.

The advent of machine learning (ML) has opened new avenues for flexible estimation of high-dimensional statistical objects, including the structural function. However, regularization and overfitting introduce bias in ML estimators, complicating both consistent estimation and valid inference for treatment effects. Debiasing is therefore indispensable, but few studies have addressed it in the panel data context. Notable exceptions include \citet{belloni2016inference}, who use a Lasso variant to select important variables and then conduct a final OLS regression of the outcome on the treatment using the selected variables as controls, yielding a consistent treatment-effect estimator under suitable conditions. \citet{klosin2022estimating,semenova2023inference,clarke2025double} instead adopt the double machine learning (DML) framework of \citet{chernozhukov2018double}, which augments the original moment condition with a debiasing term so that the resulting moment condition is less sensitive to bias in the ML estimator. With the exception of \citet{klosin2022estimating}, however, these works confine attention to partially linear models in which only the confounding effect can be nonlinear. Among them, only \citet{semenova2023inference} explicitly accommodate endogeneity in dynamic panels. Moreover, all the aforementioned works only allow for individual fixed effects, but not time fixed effects. Given the potential for omitted global confounders, it is more appropriate to include both individual and time fixed effects in the structural function.

To address these gaps, we develop a framework for estimating and conducting inference on the average derivative of the structural function in panel data with regressors that are not strictly exogenous. We consider both a static and a dynamic panel model, where the structural functions are nonparametric and incorporate additive two-way fixed effects (TWFE). These functions depend on contemporaneous and historical treatments and covariates, thereby capturing dynamic treatment effects and rich effect heterogeneity. Individual and time fixed effects are eliminated by differencing along the time and individual dimensions, respectively, and ML is then employed to estimate the differenced structural function. The use of ML is motivated by two considerations: it can flexibly approximate complex functions, and its built-in regularization helps mitigate the ill-posedness problem inherent in nonparametric estimation under endogeneity.

To debias ML estimators, we must specify a suitable debiasing term. \citet{klosin2022estimating} derive such a term in an exogenous nonparametric panel model with only individual fixed effects and estimate it by adapting the auto-DML framework of \citet{chernozhukov2022automatic}, originally designed for cross-sectional data. To our knowledge, however, no prior work has addressed debiased ML estimation in panel data with both endogenous regressors and time fixed effects---features that substantially complicate the problem. We therefore first derive the form of the debiasing term in our setting and then extend the auto-DML framework accordingly. Because endogeneity precludes the use of the machine learners considered in \citet{chernozhukov2022automatic}, we instead estimate the debiasing term via penalized GMM \citep{caner2018high}, drawing on the standard GMM approach to endogeneity in linear models. This preserves the central appeal of auto-DML---the ability to estimate the highly complex debiasing term using only known moment conditions, without solving for it explicitly. We further derive the convergence rate of the estimated debiasing term and establish the asymptotic theory for our debiased estimator of the average derivative.

We evaluate the proposed framework through a Monte Carlo study and an empirical application. The simulation results demonstrate that our estimator effectively removes the regularization bias of the underlying ML estimator, provides accurate variance estimates, and yields confidence intervals with empirical coverage near the nominal level. Moreover, it remains robust to first-stage misspecification and to the mistaken neglect of endogeneity. In the empirical application, we use the ECLS-K dataset to examine how family socioeconomic status, household size, and fast food prices shape early childhood BMI. The results show that the contemporaneous effect of family socioeconomic status changes sign as children age—a finding that helps reconcile opposing conclusions in prior studies regarding the direction of this effect.

Our paper resolves a set of challenges that, taken together, distinguish it from existing work and makes several contributions to related strands of literature. First, we extend the DML framework to nonparametric panel models with TWFE in the presence of endogeneity---a setting that, to our knowledge, has not been studied before. In this novel context, we derive the explicit form of the debiasing term and use it to construct Neyman orthogonal moment conditions. A further complication unique to TWFE panels arises in the cross-fitting step: removing time fixed effects requires demeaning across units within each period, an operation that mechanically introduces correlation across the demeaned data and thereby breaks the independence on which standard cross-fitting relies. We resolve this with an innovative cross-fitting scheme that sets aside an additional fold dedicated to demeaning, thereby restoring the sample independence required for the asymptotic theory.
 
Second, the debiasing term is essentially the influence function of the original moment condition \citep{ichimura2022influence}; we show that it factors into the product of a Riesz representer and the differenced error. The Riesz representer admits no closed-form solution in our setting. Existing estimation strategies \citep[e.g.,][]{chernozhukov2021automatic,chernozhukov2022riesznet,chernozhukov2020adversarial} are confined to exogenous regressors and cannot be transported to our setting. We extend the auto-DML idea by estimating the Riesz representer via penalized GMM, delivering an estimator that operates under both endogeneity and TWFE. Along the way, we derive the mean-squared convergence rate of penalized GMM under conditions different from---and milder than---those in \citet{caner2018high,bakhitov2022automatic}, recovering across a broad range of situations the same rate available in the cross-sectional exogenous benchmark. This also extends the classical nonparametric instrumental variables framework of \citet{newey2003instrumental} to the panel data setting.
 
Third, our framework accommodates heterogeneous and dynamic effects of continuous treatments, overcoming a well-known limitation of conventional TWFE estimators, which are largely confined to estimating contemporaneous, homogeneous effects of binary treatments, have been shown to deliver negatively-weighted aggregates of treatment effects under heterogeneity and are criticized for relying on assumptions unlikely to hold exactly \citep{de2020two,callaway2021difference,arkhangelsky2024causal}. Meanwhile, we retain both individual and time fixed effects, which helps address the issue of omitted variables and aligns with classical panel model practices \citep{mundlak1978pooling}. 
 
Fourth, dynamic panel models suffer from endogeneity due to the inclusion of lagged dependent variables, giving rise to the well-known Nickell bias. While the classical Arellano-Bond estimator resolves this issue in linear models, our framework addresses this bias based on orthogonalized Arellano-Bond-type moment conditions. Moreover, our specification does not impose linearity restrictions, thereby mitigating the misspecification risk inherent in conventional dynamic panel models and yielding a refined and extended version of the classic Arellano–Bond estimator.

The rest of this paper is organized as follows. Section \ref{sec2} introduces the framework of this paper, including notation, models, estimands, and basic assumptions. In Section \ref{sec3} we describe the estimation procedure, i.e., the algorithm to estimate the debiasing term, and then build the debiased estimator. Section \ref{sec4} discusses the mean-square convergence rate of penalized GMM and establishes the asymptotic theory of our debiased estimators. Simulation results and an empirical application regarding the determinants of early childhood BMI are presented in Sections \ref{sec5} and \ref{sec6}, and Section \ref{sec7} concludes. Additional theoretical and simulation results are provided in Appendices A–C, and the proofs of all results are available in the supplemental appendix.

\section{Framework Setup}
\label{sec2}
Under the cross-sectional setting, \citet{newey2003instrumental} and \citet{newey2013nonparametric} have studied the endogeneity problem through the classical nonparametric instrumental variables (NPIV) framework
\begin{equation}
  \label{ivm}
  Y = \gamma_0(X) + \varepsilon, \quad \E(\varepsilon \mid Z) = 0,
\end{equation}
where $Y$ is the outcome,  $X$ is the regressor and is correlated with the error term $\varepsilon$, and $Z$ is the instrumental variable. In this model and the models below, we slightly abuse notation by using $\gamma_0(\cdot)$ to denote an unknown function with the appropriate input dimension. We will extend this model to the panel data setting to study the causal effects of a continuous treatment. 

Specifically, we consider a balanced panel dataset\footnote{The extension to unbalanced panels is conceptually straightforward and requires only technical modifications including applying the differencing, fold-specific demeaning, and debiased moments to observed period pairs, and setting unavailable moment contributions to zero.} with $N$ individuals and $T$ periods, indexed by $i=1,\ldots,N$ and $t=1,\ldots,T$, respectively. For individual $i$ in period $t$, let $Y_{it}$ denote the scalar outcome, $D_{it}$ the scalar continuous treatment, $X_{it}$ the vector of time-varying covariates, and $I_{it}$ the vector of time-varying instrumental variables. $C_i$ is the vector of time-invariant covariates. In the models below, $\mu_i$ and $\lambda_t$ represent the individual and time fixed effects, respectively, and $\varepsilon_{it}$ is the idiosyncratic error term. To simplify the notation for the exposition, we define $\overline{V}_{it} = (V_{i1},\ldots,V_{it})^\prime$, $\overline{W}_t = (W_1,\ldots,W_t)^\prime$, $\underline{V}_{it} = (V_{it},\ldots,V_{iT})^\prime$, and $\overline{V}_{i,j:k} = (V_{ij},V_{i,j+1},\ldots,V_{iT})^\prime$ for some variables $V_{it}$ and $W_t$. Following the literature, we assume that the data $(\overline{D}_{iT},\overline{X}_{iT},\overline{I}_{iT},\overline{\varepsilon}_{iT},C_i,\mu_i)_{i=1}^N$ are independent and identically distributed for different individuals, which is referred to as the i.i.d. assumption.

Endogeneity in panel data manifests in various forms. In static panel models, time-varying covariates may be sequentially exogenous, meaning they are correlated with future error terms, or completely endogenous, meaning they are correlated with the error terms in all periods. These scenarios lead to the first model studied in this paper, which is a nonparametric static panel model with additive two-way fixed effects: 
\begin{equation}
  \label{mod1}
  \begin{gathered}
    Y_{it} = \gamma_0(\overline{X}_{i,t-q:t},\overline{D}_{i,t-q:t},C_i) + \mu_i + \lambda_t + \varepsilon_{it}, \\
    \E(\varepsilon_{it} \mid \overline{X}_{1it},\overline{D}_{it},\overline{I}_{it},C_i,\mu_i,\lambda_t) = 0,
  \end{gathered}
\end{equation}
where $X_{1it}$ is a subvector of $X_{it}$. The integer $q \geq 0$ is specified by the user, indicating that treatments and time-varying covariates between period $t-q$ and $t$ may affect the outcome via $\gamma_0(\cdot)$. Since some time-varying covariates may respond dynamically to historical outcome variables, we assume they are sequentially exogenous rather than strictly exogenous, as indicated by the mean independence between $\overline{X}_{1it}$ and $\varepsilon_{it}$. Meanwhile, some time-varying covariates might be completely endogenous, and there exist instruments $\overline{I}_{it}$ that are mean independent of $\varepsilon_{it}$, so we may utilize these instruments for the identification and estimation. Moreover, it should be noted that \eqref{mod1} is a flexible model as it does not impose any functional restrictions on the treatment effect and the confounding effect of covariates. 

In dynamic panel models, the lagged outcome variables, as part of the regressors, are bound to be correlated with historical error terms, making the regressors endogenous. Moreover, as in static panel models, the time-varying covariates can be sequentially exogenous or completely endogenous. Hence, we consider the following nonparametric dynamic panel model with additive two-way fixed effects, 
\begin{equation}
  \label{mod2}
  \begin{gathered}
    Y_{it} = \gamma_0(\overline{Y}_{i,t-1-p:t-1},\overline{X}_{i,t-q:t},\overline{D}_{i,t-q:t},C_i) + \mu_i + \lambda_t + \varepsilon_{it},\\ 
    \E(\varepsilon_{it} \mid \overline{Y}_{i,t-1},\overline{X}_{1it},\overline{D}_{it},\overline{I}_{it},C_i,\mu_i,\lambda_t) = 0,
  \end{gathered}
\end{equation}
where the integer $p \geq 0$ is selected by the user, stipulating the order of autoregressive effects of $Y_{it}$. The interpretation of \eqref{mod2} is similar to that of \eqref{mod1} and is thus omitted for simplicity. 

We now turn to the discussion of causal estimands. In practice, the unknown function itself is rarely an object of interest, the parameter of interest would typically be a functional of $\gamma_0(\cdot)$ of the form  
\[
\theta_0 = \E[g(W_{it},\gamma_0)], 
\]
which is an expectation of some functional $g(W_{it},\gamma_0)$ over the distribution of data, where $W_{it} = (\overline{Y}_{it},\overline{D}_{it},\overline{X}_{it},C_i)^\prime$. Following this, for $s\in\{0,1,\ldots,q\}$, we will estimate the average causal response of the form  
\[
\theta_{0t}(s) = \E\left[\frac{\partial Y_{it}}{\partial D_{i,t-s}}\right] = \E\left[\frac{\partial \gamma_0}{\partial D_{i,t-s}}\right],
\]
where $\gamma_0 = \gamma_0(\overline{X}_{i,t-q:t},\overline{D}_{i,t-q:t},C_i)$ under the static panel model, and $\gamma_0 = \gamma_0(\overline{Y}_{i,t-1-p:t-1}, \allowbreak \overline{X}_{i,t-q:t},\overline{D}_{i,t-q:t},C_i)$ under the dynamic panel model. In the cross-sectional setting or panel setting without dynamic treatment effects, $\theta_{0t}(0)$ is a very natural estimand when treatments are continuous and has been widely studied in many existing works \citep{ai2007estimation,callaway2024difference,klosin2022estimating,xiao2025causal}. It measures the average marginal effect of the contemporaneous treatment on the current outcome (period $t$'s outcome), which is as fundamental as the classical average treatment effect estimand under the binary treatment scenario. Because our framework is designed for panel data allowing dynamic treatment effects, we extend $\theta_{0t}(0)$ to $\theta_{0t}(s)$ for $s\in\{0,1,\ldots,q\}$, so we are able to explore the average marginal effect of historical treatments on the current outcome. 

After defining $\theta_{0t}(s)$, it is reasonable to consider weighted averages of $\theta_{0t}(s)$'s that measure certain aggregated effects. For example, we may consider the aggregated average marginal effect of contemporaneous and historical treatments on the current outcome, 
$
\theta_{0t} = \sum_{s=0}^q w_s \theta_{0t}(s),
$
and the aggregated average marginal effect of contemporaneous treatment across periods, 
$
\theta_{0}(0) = \sum_{t=1}^T w_t \theta_{0t}(0),
$
where $w_s$'s and $w_t$'s are user-defined weights.

\section{Estimation}
\label{sec3}
\subsection{Overview}
\label{sec31}
In the cross-sectional setting, the nonparametric instrumental regression model \eqref{ivm} has been extensively studied in the existing literature, and readers can refer to \citet{horowitz2011applied,chen2016methods,centorrino2017additive} for some reviews. The most prominent challenge in the estimation of this model is the ill-posedness problem. Note that the conditional expectation condition in \eqref{ivm} yields the integral equation 
\[
\E[Y \mid Z] = \int \gamma_0(x) f(x \mid z) \rd x,
\]
where $f$ represents the conditional PDF of $X$ given $Z$. This is an integral Fredholm equation of the first kind, and solving for $\gamma_0$ directly is ill-posed because it involves inverting a compact operator, leading to the mapping from $\E[Y \mid Z]$ to $\gamma_0$ being not continuous. This brings about severe estimation issues. For example, a small estimation error of the conditional expectation can cause a huge fluctuation in the solution of $\gamma_0$, indicating that we cannot simply plug estimators of $\E[Y \mid Z]$ and $f(x \mid z)$ into the integral equation and then approximately solve $\gamma_0$. Another issue arising therefrom is that the estimation of $\gamma_0$ exhibits a slower convergence rate than the standard nonparametric regression. 

A common solution to the ill-posedness problem is regularization, which leverages a penalty term to avoid estimation of higher-order terms that induce variance explosion. For example, \citet{florens2003inverse,carrasco2000generalization,hall2005nonparametric,carrasco2007linear} use the Tikhonov regularization. \citet{newey2003instrumental} resort to a different method based on sieve estimation under regularization, \citet{florens2016regularizing} consider the Bayesian regularization approach, and \citet{singh2019kernel} studies the reproducing kernel Hilbert spaces regularization.

Noting the critical role of regularization in machine learning (ML), recent literature has attempted to employ ML to tackle the estimation of $\gamma_0$. ML methods enjoy the following advantages over traditional methods. First, compared with kernel regression or sieve approximation, machine learner like neural networks, random forests, and boosting methods can provide more flexible functional form approximation for the conditional expectations $\E[Y\mid Z]$ and $\E[X\mid Z]$ that may appear in the estimation procedure. Second, machine learners exploit sophisticated regularization schemes that are data-driven to control the complexity of ML estimators and simultaneously alleviate the ill-posedness problem. Third, when the instrument $Z$ is high-dimensional, the convergence rate under traditional regularization decreases, which is a manifestation of the curse of dimensionality. In contrast, machine learners can more effectively extract information from the data, thereby constructing more powerful instruments. The implementation of ML estimators is beyond the scope of this article; interested readers can refer to the references in the review of ML instrumental variable estimators of \citet{bakhitov2022automatic} for further details. 

Suppose we obtain an ML estimator $\widehat{\gamma}$ and the parameter of interest $\beta_0$ satisfies the moment condition $\E[m(Y,X,Z,\gamma_0) - \beta_0] = 0$. By plugging $\widehat{\gamma}$ into the empirical version of the moment condition, we obtain a naive estimator $\widehat{\beta} = \sum_{i=1}^N m(Y_i,X_i,Z_i,\widehat{\gamma})/N$. Nonetheless, as we pointed out before, ML estimators are built under regularization, which would typically generate biased estimators. In the literature, the bias introduced by regularization is called regularization bias. Worse still, ML estimators generally suffer from another serious issue -- overfitting, which thereby gives rise to overfitting bias. Hence, the ML estimator $\widehat{\gamma}$ is typically accompanied by non-negligible biases, rendering the plug-in estimator $\widehat{\beta}$ also biased and not $\sqrt{N}$-consistent, and the asymptotic inference invalid \citep{chernozhukov2017double,chernozhukov2018double,chernozhukov2022locally}. Hence, \citet{chernozhukov2018double} propose the DML framework to build debiased estimators that are $\sqrt{N}$-consistent and asymptotically normal under mild regularity conditions. The core argument of this framework is that the regularization bias and overfitting bias can be removed by Neyman orthogonal moment conditions and sample splitting, respectively.

Briefly put, the original moment condition $\E[m(Y,X,Z,\gamma_0) - \beta_0] = 0$ is sensitive to the disturbance in $\gamma_0$ as its Gateaux derivative with respect to $\gamma_0$ is not zero. Conversely, by adding a debiasing term to the original moment condition, the Gateaux derivative of the Neyman orthogonal moment condition vanishes, so the bias of $\widehat{\gamma}$ has no first-order influence on the debiased estimator. Consequently, the negative impact of regularization bias can be greatly reduced. 

The overfitting bias refers to that if the same data are used to construct the function estimator $\widehat{\gamma}$ and then calculate the debiased estimator for the parameter of interest, certain remainders in the asymptotic analysis for the debiased estimator will not vanish, giving rise to non-negligible biases. \citet{chernozhukov2018double} provide illustrative examples of the overfitting bias. However, if data of some individuals are used to construct $\widehat{\gamma}$, and data of different individuals are utilized to estimate the parameter of interest, then those remainders will become asymptotically negligible under the i.i.d. assumption, thereby eliminating the overfitting bias. Therefore, we first split the sample along the individual dimension into two parts, use data in the first part to build $\widehat{\gamma}$, and finally use data in the second part to construct the debiased estimator. This procedure is termed sample splitting. Nevertheless, sample splitting may result in efficiency loss because only part of the data are used in the final estimation step. Thus, in the estimation algorithm we adopt an upgraded version of sample splitting called cross-fitting, where every part of the sample will be used in the final estimation step to recover full efficiency.

When we extend the cross-sectional NPIV model to panel NPIV models like \eqref{mod1} and \eqref{mod2}, the ill-posedness problem continues to exist. Hence, machine learners with regularization remain powerful tools for the construction of good estimators for $\gamma_0$. Undoubtedly, like the cross-sectional case, plugging the ML estimator $\widehat{\gamma}$ into the panel moment condition would give an estimator of interest parameter that is inconsistent and not asymptotically normal. However, the DML framework has not been adequately studied in the panel case, and no prior work has considered the application of DML to nonparametric panel data models with two-way fixed effects in the presence of endogeneity to derive debiased estimators for functional parameters of interest. This is the problem we address below. It is noteworthy that a big challenge therein lies in the estimation of an additional function -- contained in the debiasing term for the construction of the Neyman orthogonal moment condition -- that generally has no analytical solution. Thus, drawing on existing work, we adopt an approach that relies solely on the moment conditions for the estimation of $\theta_{0t}(s)$ and does not require any information about the function itself to estimate this additional function. Next, we present the technical details for the estimation part of our framework. 

\subsection{Estimation Algorithm}
In this section, we first discuss detailed steps for the estimation of the estimand $\theta_{0t}(s)$ under the static panel model \eqref{mod1}. 

We commence with eliminating the individual fixed effect by differencing. Specifically, we subtract the model for period $t-s-1$ from the model for period $t$ to obtain  
\begin{equation}
  \label{fd}
  \Delta Y_{it} = \Delta \gamma_0(\overline{X}_{i,t-q:t},\overline{D}_{i,t-q:t},C_i) + \Delta \lambda_t + \Delta \varepsilon_{it},
\end{equation}
where $\Delta Y_{it} = Y_{it} - Y_{i,t-s-1}$, $\Delta \lambda_t = \lambda_t - \lambda_{t-s-1}$, $\Delta \varepsilon_{it} = \varepsilon_{it} - \varepsilon_{i,t-s-1}$, $\Delta \gamma_0(\overline{X}_{i,t-q:t}, \allowbreak \overline{D}_{i,t-q:t},C_i) \allowbreak = \gamma_0(\overline{X}_{i,t-q:t},\overline{D}_{i,t-q:t},C_i) - \gamma_0(\overline{X}_{i,t-q-s-1:t-s-1},\overline{D}_{i,t-q-s-1:t-s-1},C_i)$. \\Note that the difference operator $\Delta$ actually depends on $s$, but we omit $s$ for the sake of notational simplicity. 

Then, we randomly split $N$ individuals into $K$ folds $\{\mathcal{F}_1,\ldots,\mathcal{F}_K\}$ such that each individual is in one and only one fold. Next, for a $|A|$-element set $A = \{a_1,a_2,\ldots,a_{|A|}\} \subseteq \{1,\ldots,K\}$ with $a_1 < a_2 < \cdots < a_{|A|}$, we define a mapping $\pi(\cdot,A): A \mapsto A$ that selects another fold for each fold in $A$,
\[
\pi(a_j,A) = \left\{\begin{array}{ll}
  a_{j+1} & \text{if } j < |A|,\\
  a_{1} & \text{if } j = |A|.\\
\end{array}\right.
\]
For the fold $\mathcal{F}_k$, we choose another fold $\mathcal{F}_{k^\prime}$, where $k^\prime = \pi(k,\{1,\ldots,K\})$, and then subtract the cross-sectional mean of fold $\mathcal{F}_{k^\prime}$ from each term in \eqref{fd} to eliminate the time fixed effect, 
\[
  \Delta Y_{it}^* = \Delta Y_{it} - \frac{1}{|\mathcal{F}_{k^\prime}|}\sum_{j\in\mathcal{F}_{k^\prime}} \Delta Y_{jt} = \Delta \gamma_0^*(\overline{X}_{i,t-q:t},\overline{D}_{i,t-q:t},C_i) + \Delta \varepsilon_{it}^*,
\]
where $i \in \mathcal{F}_k$ (this is implicitly presumed in the following demonstration), 
\begin{gather*}
  \Delta \gamma_0^*(\overline{X}_{i,t-q:t},\overline{D}_{i,t-q:t},C_i) = \Delta \gamma_0(\overline{X}_{i,t-q:t},\overline{D}_{i,t-q:t}, C_i) \\
  - \frac{1}{|\mathcal{F}_{k^\prime}|}\sum_{j\in\mathcal{F}_{k^\prime}} \Delta \gamma_0(\overline{X}_{j,t-q:t},\overline{D}_{j,t-q:t},C_j), \text{ and }  \Delta \varepsilon_{it}^* = \Delta \varepsilon_{it} - \frac{1}{|\mathcal{F}_{k^\prime}|}\sum_{j\in\mathcal{F}_{k^\prime}} \Delta \varepsilon_{jt}.
\end{gather*}
Regarding the mean independence of the error term $\Delta \varepsilon_{it}^*$, it follows from $\E(\varepsilon_{it} \mid \overline{X}_{1it},\overline{D}_{it},\allowbreak \overline{I}_{it},C_i,\mu_i,\lambda_t) = 0$ that $\E(\Delta \varepsilon_{it} \mid \overline{X}_{1,i,t-s-1},\overline{D}_{i,t-s-1},\overline{I}_{i,t-s-1},C_i) = 0$. Therefore, by the i.i.d. assumption, we have 
\begin{equation}
  \label{eqer}
  \E(\Delta \varepsilon_{it}^* \mid \overline{X}_{1,i,t-s-1},\overline{D}_{i,t-s-1},\overline{I}_{i,t-s-1},C_i) = 0.
\end{equation}

By our definition of $\Delta \gamma^*_0$, it is easy to verify that 
\begin{equation}
  \label{eqraw}
  \theta_{0t}(s) = \E\left[\frac{\partial \gamma_0(\overline{X}_{i,t-q:t},\overline{D}_{i,t-q:t},C_i)}{\partial D_{i,t-s}}\right] = \E\left[\frac{\partial \Delta \gamma_0^*(\overline{X}_{i,t-q:t},\overline{D}_{i,t-q:t},C_i)}{\partial D_{i,t-s}}\right].
\end{equation}
Therefore, defining $W_{it}^* = (W_{jt})_{j\in\{i\}\cup\mathcal{F}_{k^\prime}}$ and $\displaystyle m(W_{it}^*,\Delta \gamma_0^*) = \frac{\partial \Delta \gamma_0^*(\overline{X}_{i,t-q:t},\overline{D}_{i,t-q:t},C_i)}{\partial D_{i,t-s}}$, if we have a nuisance function estimator $\widehat{\gamma}$ for $\gamma_0$, we can construct a plug-in estimator for $\theta_{0t}(s)$ by the sample average 
\[
\widehat{\theta}_{0t}^p(s) = \frac{1}{N}\sum_{i=1}^N m(W_{it}^*,\Delta \widehat{\gamma}^*).
\]
Nonetheless, as we mentioned before, the plug-in estimator is sensitive to the first-order bias in the estimator $\widehat{\gamma}$. If the first-order bias does not vanish at $\sqrt{N}$-rate, which is typically the case if we use ML to build $\widehat{\gamma}$, the resulting estimator $\widehat{\theta}_{0t}^p(s)$ is not consistent and the asymptotic inference is not valid as well. Consequently, to establish the asymptotic theory, we will add a debiasing term to $\widehat{\theta}_{0t}^p(s)$ to counteract the influence of the first-order bias. 

Let $\gamma(\cdot)$ denote a function of $\overline{X}_{i,t-q:t},\overline{D}_{i,t-q:t},C_i$ different from the true value $\gamma_0(\cdot)$, we observe that generally the Gateaux derivative of $\E[m(\cdot)]$ in the direction $\gamma$ is not zero, i.e., 
\[
\frac{\partial \E[m(W_{it}^*,\Delta \gamma_0^* + r\Delta \gamma^*)]}{\partial r}\bigg|_{r=0} = \E[m(W_{it}^*,\Delta\gamma^*)] \neq 0,
\]
which verifies the existence of the influence of the first-order bias. 

Define $V_{it} = (\overline{X}_{i,t-q:t},\overline{D}_{i,t-q:t},C_i)^\prime$ and $Z_{it} = (\overline{X}_{1,i,t-s-1},\overline{D}_{i,t-s-1},\overline{I}_{i,t-s-1},C_i)^\prime$. To construct orthogonal conditions, we assume there is a function $\alpha_0(Z_{it})$ with $\E[\alpha_0(Z_{it})]^2 < \infty$ such that 
\begin{equation}
  \E[\partial h(V_{it},V_{i,t-s-1})/\partial D_{i,t-s}] = \E[\alpha_0(Z_{it})h(V_{it},V_{i,t-s-1})]
  \label{eqrie0}
\end{equation}
for all $h$ such that $\E[h^2] < \infty$ and the derivative and its expectation exist. Setting $h \equiv 1$, we have $\E[\alpha_0(Z_{it})]=0$. Thus, under the i.i.d. assumption and taking $h$ as $\Delta \gamma$, \eqref{eqrie0} and our cross-fold demeaning scheme would imply 
\begin{equation}
  \label{eqrie}
  \E[m(W_{it}^*,\Delta \gamma^*)] = \E[\alpha_0(Z_{it})\Delta \gamma^*(V_{it})]
\end{equation}
for all $\Delta \gamma$ such that $\E[\Delta \gamma^*(V_{it})]^2 < \infty$ and $\E[m(W_{it}^*,\Delta \gamma^*)]$ exists.

By the Riesz representation theorem, if $\E[\partial \Delta \gamma(V_{it})/\partial D_{i,t-s}]$ is a mean square continuous functional of $\Delta \gamma$, i.e, $\E[\partial \Delta \gamma(V_{it})/\partial D_{i,t-s}] < C\sqrt{\E[\Delta \gamma(V_{it})]^2}$ for all $\Delta \gamma$, then there exists $g(V_{it},V_{i,t-s-1})$ with $\E[g(V_{it},V_{i,t-s-1})]^2 < \infty$ such that $\E[\partial \Delta \gamma(V_{it})/\partial D_{i,t-s}] = \E[g(V_{it},V_{i,t-s-1}) \allowbreak \Delta \gamma(V_{it})]$. In this case, the existence of $\alpha_0(Z_{it})$ satisfying $g(V_{it},V_{i,t-s-1}) \allowbreak = \E[\alpha_0(Z_{it}) \mid V_{it},V_{i,t-s-1}]$ would imply the existence of $\alpha_0(Z_{it})$ satisfying \eqref{eqrie}. It is noteworthy that requiring $g(V_{it},V_{i,t-s-1}) = \E[\alpha_0(Z_{it}) \mid V_{it},V_{i,t-s-1}]$ is essentially the same as Assumption 5 of \citet{ichimura2022influence}, which is a standard assumption for the construction of orthogonal conditions in the presence of endogeneity and is necessary for the existence of a $\sqrt{N}$-consistent estimator of $\theta_{0t}(s)$ as discussed in \citet{severini2012efficiency}. 

By adding the debiasing term $\alpha_0(Z_{it})[\Delta Y_{it}^* - \Delta \gamma_0^*(V_{it})]$ to the original moment condition \eqref{eqraw}, we obtain the orthogonal moment condition 
\begin{equation}
  \label{eqor}
  \E[\psi(W_{it}^*,\theta_{0t}(s),\Delta \gamma^*_0,\alpha_0)] = 0,
\end{equation}
where $\psi(W_{it}^*,\theta,\Delta \gamma^*,\alpha) = m(W_{it}^*,\Delta \gamma^*) - \theta + \alpha(Z_{it})[\Delta Y_{it}^* - \Delta \gamma^*(V_{it})]$. It follows from \eqref{eqer} that $\E[\alpha_0(Z_{it})[\Delta Y_{it}^* - \Delta \gamma^*_0(V_{it})]] = \E[\alpha_0(Z_{it}) \E[\Delta \varepsilon_{it}^* \mid Z_{it}]] = 0$, so \eqref{eqor} is a valid moment condition by means of which we can estimate $\theta_{0t}(s)$. Moreover, the Gateaux derivative of $\E[\psi(\cdot)]$ in the direction of $(\Delta \gamma^*, \alpha)$ is zero, 
\[
\begin{aligned}
  & \frac{\partial \E[\psi(W_{it}^*,\theta_{0t}(s),\Delta \gamma^*_0 + r \Delta \gamma^*,\alpha_0+r\alpha)]}{\partial r}\bigg|_{r=0} \\
  & = \E[m(W_{it}^*,\Delta \gamma^*)] - \E[\alpha_0(Z_{it})\Delta \gamma^*(V_{it})] + \E[\alpha(Z_{it})[\Delta Y_{it}^* - \Delta \gamma_0^*(V_{it})]] = 0,
\end{aligned}
\]
where the first two terms cancel each other out because of \eqref{eqrie}, and the last term equals zero because of \eqref{eqer}. Therefore, the bias of nuisance function estimators $\Delta \widehat{\gamma}_0^*$ and $\widehat{\alpha}_0$ has no first-order influence on the orthogonal moment condition, making it possible to build a debiased estimator of $\theta_{0t}(s)$ based on \eqref{eqor}.

The remaining problem is how to estimate the nuisance function $\alpha_0(Z_{it})$. As discussed in \citet{bakhitov2022automatic}, in the nonparametric instrumental variables setting, $\alpha_0(Z_{it})$ is very hard to derive or even unknown, so we had better estimate $\alpha_0(Z_{it})$ by a general and flexible approach. \citet{chernozhukov2022locally} and \citet{chernozhukov2022automatic} propose such an approach which they term ``automatic estimation'' under the cross-sectional setting with exogenous regressors, and \citet{bakhitov2022automatic} extends it to the endogenous regressors scenario. In the following, we extend their approaches to our setting. 

Note that the Gateaux derivative of $\E[\psi(\cdot)]$ in the direction of $\Delta \gamma^*$ is 
\begin{equation}
  \label{eqpop}
  \frac{\partial \E[\psi(W_{it}^*,\theta_{0t}(s),\Delta \gamma^*_0 + r \Delta \gamma^*,\alpha_0)]}{\partial r}\bigg|_{r=0} = \E[m(W_{it}^*,\Delta \gamma^*)] - \E[\alpha_0(Z_{it})\Delta \gamma^*(V_{it})] = 0.
\end{equation}
As pointed out by \citet{chernozhukov2022locally}, this can be viewed as a population moment condition for the estimation of $\alpha_0(Z_{it})$. Specifically, we first replace $\alpha_0(Z_{it})$ by a sieve and then estimate the sieve parameters using the sample analogue of \eqref{eqpop} with different choices of $\Delta \gamma^*$. Hence, we assume $\alpha_0(Z_{it})$ can be approximated by $b(Z_{it})^\prime \rho$, where $b(Z_{it}) = (b_1(Z_{it}),\ldots,b_p(Z_{it}))^\prime$ is a $p$-dimensional dictionary containing functions of $Z_{it}$, and we also prepare a $q$-dimensional dictionary $d(V_{it}) = (d_1(V_{it}),\ldots,d_q(V_{it}))^\prime$ consisting of functions of $V_{it}$. In practice, we will standardize $b(Z_{it})$ and $d(V_{it})$ to achieve better estimation performance. Then, for the $q$ functions $d_1(V_{it})$ to $d_q(V_{it})$, we have $q$ moment conditions, 
\begin{equation}
  \label{eqsam}
  \E[m(W_{it}^*, \Delta d_j^*) - \Delta d_j^*(V_{it})b(Z_{it})^\prime\rho] = 0, \quad j=1,\ldots q,
\end{equation}
where $\Delta d_j^*$ is defined in step 2(b) of Algorithm \ref{ag1}. Note that although $\Delta d_j^*$ is defined in a manner similar to $\Delta \gamma_0^*$, the fold used to subtract the cross-sectional mean is selected in a different way to ensure the data for constructing $\widehat{\alpha}_k$ are independent of the data for the estimation of $\theta_{0t}(s)$, complying with the core principle of cross-fitting. Under the condition that $q \geq p$, it is straightforward to obtain a standard GMM estimator $\widehat{\rho}$ based on $q$ moment conditions in \eqref{eqsam}. However, sometimes $\alpha_0(Z_{it})$ is high-dimensional in the sense that $p > n$. To allow for such a high-dimensional case, we may resort to the regularization technique. To this end, we follow the high-dimensional linear GMM framework of \citet{caner2018high} to construct a penalized GMM estimator for $\rho$, denoted by $\widehat{\rho}$. Then, $\widehat{\alpha}_0(Z_{it}) = b(Z_{it})^\prime \widehat{\rho}$ is our estimator for $\alpha_0(Z_{it})$ and a debiased estimator for $\theta_{0t}(s)$ can be easily built based on \eqref{eqor}. We summarize all the calculation details in the following algorithm, where the cross-fitting technique is incorporated as well to facilitate the development of the asymptotic theory. 
\begin{algorithm}[Debiased estimator of $\theta_{0t}(s)$] \label{ag1} \ 
\begin{enumerate}[leftmargin=0.5cm]
  \item Randomly split $N$ individuals $\{1,\ldots,N\}$ into $K$ folds $\{\mathcal{F}_1,\ldots,\mathcal{F}_K\}$ of about equal size such that each individual is in one and only one fold. The size of fold $\mathcal{F}_k$ is $n_k$. 
  \item For each $k \in \{1,\ldots,K\}$:
  \begin{enumerate}[leftmargin=0.5cm]
    \item Choose another fold $\mathcal{F}_{k^\prime}$, where $k^\prime = \pi(k,\{1,\ldots,K\})$, and define \\$\mathcal{F}_{k,k^\prime}^c = \{1,\ldots,N\} \backslash (\mathcal{F}_k\cup\mathcal{F}_{k^\prime})$. 
    \item Construct (machine learning) estimators $\widehat{\gamma}_k$ and $\widehat{\alpha}_k$ using data of individuals in $\mathcal{F}_{k,k^\prime}^c$. To be specific, $\widehat{\gamma}_k$ can be constructed by, for example, penalized GMM or DeepIV \citep{hartford2017deep}. Subsequently, for each $k^* \in \{1,\ldots,K\}\backslash\{k,k^\prime\}$, define $k^{*\prime} = \pi(k^*,\{1,\ldots,K\}\backslash\{k,k^\prime\})$, and for each $i \in \mathcal{F}_{k^*}$, define 
    \[
    \Delta d_j^*(V_{it}) = \Delta d_j(V_{it}) - \frac{1}{n_{k^{*\prime}}} \sum_{m\in\mathcal{F}_{k^{*\prime}}} \Delta d_j(V_{mt}), \quad j=1,\ldots,q,
    \]
    and $\Delta d^*(V_{it}) = (\Delta d_1^*(V_{it}),\ldots,\Delta d_q^*(V_{it}))^\prime$. We further define  
    \begin{gather}
      \widehat{G}_k = \frac{1}{N-n_k-n_{k^\prime}} \sum_{i\in\mathcal{F}_{k,k^\prime}^c} \Delta d^*(V_{it})b(Z_{it})^\prime, \notag \\
      \widehat{M}_k = \frac{1}{N-n_k-n_{k^\prime}}\sum_{i\in\mathcal{F}_{k,k^\prime}^c} (m(W_{it}^*,\Delta d_1^*),\ldots,m(W_{it}^*,\Delta d_q^*))^\prime, \notag \\
      \text{and } \widehat{\rho}_k = \mathop{\arg\min}_{\rho\in\mathbb{R}^p} \big(\widehat{M}_k - \widehat{G}_k\rho\big)^\prime \widehat{\Omega} \big(\widehat{M}_k - \widehat{G}_k\rho\big) + 2 r \|\rho\|_1, \label{eqop}
    \end{gather}
    where \eqref{eqop} is the penalized sample analogue of \eqref{eqsam} with the penalty term $2 r \|\rho\|_1$. Then, $\widehat{\alpha}_k$ is obtained by $\widehat{\alpha}_k = b(Z_{it})^\prime \widehat{\rho}_k$.
  \end{enumerate}
  \item Plug $\widehat{\gamma}_k$ and $\widehat{\alpha}_k$ of each fold into the sample analogue of \eqref{eqor} to construct the debiased estimator for $\theta_{0t}(s)$, i.e.,
  \[
  \widehat{\theta}_{0t}^d(s) = \frac{1}{N} \sum_{k=1}^K \sum_{i\in\mathcal{F}_k} m(W_{it}^*,\Delta \widehat{\gamma}_k^*) + \widehat{\alpha}_k(Z_{it})[\Delta Y_{it}^* - \Delta \widehat{\gamma}_k^*(V_{it})].
  \]
\end{enumerate}
\end{algorithm}

Based on $\widehat{\theta}_{0t}^d(s)$, it is straightforward to build debiased estimators for $\theta_{0t}$ and $\theta_0(0)$, as demonstrated in Appendix \ref{appa}.

Next, for the dynamic panel model \eqref{mod2}, the algorithm for the construction of $\widehat{\theta}_{0t}^d(s)$ would be very similar to the above one, and the debiased estimators for $\theta_{0t}$ and $\theta_0(0)$ are easily obtained as well. To save space, we relegate the technical details to Appendix \ref{appb}.

\section{Asymptotic Theory}
\label{sec4}
We develop the asymptotic theory for the debiased estimator $\widehat{\theta}_{0t}^d(s)$ under fixed $T$ and $N \to \infty$. \citet{chernozhukov2022locally} have developed the asymptotic theory for the debiased estimator under the cross-sectional data setting. In this section, we extend their theory to the panel data setting with two-way fixed effects, and explicitly take into account the impact of endogeneity. It can be shown that the asymptotic behavior of $\widehat{\theta}_{0t}^d(s)$ depends on the convergence rate of the function estimators $\widehat{\gamma}_k$ and $\widehat{\alpha}_k$. The desired convergence rate of $\widehat{\gamma}_k$ is not difficult to obtain as we may use common machine learners to build $\widehat{\gamma}_k$, while the fast convergence rate of many common machine learners has been widely established. The convergence rate of $\widehat{\alpha}_k$ is more challenging to establish, and we focus on this issue in what follows.

Let $\{\varepsilon_N\}$ denote a sequence that converges to zero no faster than $\sqrt{\ln(q)/N}$. We now introduce some assumptions for the mean square convergence rate of $\widehat{\alpha}_k$. In the following, for a vector $x \in \mathbb{R}^p$, let $\|x\|_1$, $\|x\|_2$, $\|x\|_\infty$ denote the $L_1$, $L_2$, $L_\infty$ norm of $x$, respectively. For a matrix $A$, we define $\|A\|_{\infty} = \max_{i,j}|A_{ij}|$. And for a random variable $X$, let $\|X\| = \sqrt{\E(X^2)}$. For an index set $J \subset \{1,\ldots,p\}$ and a vector $x \in \mathbb{R}^p$, we define $x_J = (x_i)_{i\in J}$ as the subvector of $x$ whose elements have indices in $J$. $J^c$ is the complementary set of $J$, i.e., $J^c = \{1,\ldots,p\} \backslash J$.
\begin{assumption}
  \label{asscp}
  There exists a positive definite matrix $\Omega$ such that $\|\widehat{\Omega} - \Omega\|_\infty = O_p(\varepsilon_N)$. In addition, $\lambda_{max}(\Omega) < C_\Omega < \infty$ uniformly in $N$, where $\lambda_{max}(\Omega)$ is the largest eigenvalue of $\Omega$.
\end{assumption}
\noindent This is a standard assumption in the GMM literature. In general situations, the convergence rate for $\widehat{\Omega}$ can be satisfied. The bounded largest eigenvalue assumption is also very common, and it is noteworthy that bounded largest eigenvalue would imply $\|\Omega\|_\infty = O(1)$. 

\begin{assumption}
  \label{asscs}
  There exists $C_1>0$ and $\xi > 0$ such that for each positive integer $m \leq C_1 \varepsilon_N^{-2/(2\xi+1)}$, there is $\rho$ with no more than $m$ nonzero elements such that 
  \[
  \|\alpha_0(Z_{it}) - b(Z_{it})^\prime \rho\|^2 \leq C_1 m^{-2\xi}.
  \]
\end{assumption}
\noindent This assumption controls the mean square approximation error from using the linear combination $b(Z_{it})^\prime \rho$ to approximate $\alpha_0(Z_{it})$. As commented by \citet{chernozhukov2022automatic}, it does not require $\alpha_0(Z_{it})$ to be equal to some sparse linear combination, i.e., $\alpha_0(Z_{it})$ does not have to be strictly sparse. However, it does not rule out the scenario where $\alpha_0(Z_{it})$ is exactly a sparse linear combination, which makes the assumption general enough to accommodate many settings. Moreover, there is no need to know which elements in $b(Z_{it})$ that controls the approximation error. Thus, it allows for a very high-dimensional $b(Z_{it})$ when we believe relatively few important regressors that we may fail to know their identity would give a good approximation. 

Next, note that $\widehat{\rho}_k$ in \eqref{eqop} is defined by the sample averages $\widehat{G}_k$ and $\widehat{M}_k$, so we put forward the following assumption to control the convergence rate of these sample averages. It would imply that $\widehat{G}_k$ and $\widehat{M}_k$ will converge to their expectations at a rate of $O_p(\varepsilon_N)$. 
\begin{assumption}
  \label{assbd}
  There are constants $C_b$, $C_d$, and $C_m$ such that 
  \[
  \max_{1\leq j \leq p} |b_j(Z_{it})| \leq C_b, \ \max_{1\leq j \leq q} |d_j(V_{it})| \leq C_d, \text{ and } \max_{1\leq j\leq q}\Big|\frac{\partial d_{j}(V_{it})}{\partial D_{i,t-s}}\Big| \leq C_m,
  \]
  almost surely.
\end{assumption}

Let $\lambda_{max}(B)$ denote the largest eigenvalue of $B = \E[b(Z_{it})b(Z_{it})^\prime]$, and let $\lambda_{min}(G^\prime \Omega G)$ and $\lambda_{max}(G^\prime \Omega G)$ denote the smallest and largest eigenvalue of $G^\prime \Omega G$, respectively, where $G = \E[\Delta d^*(V_{it})b(Z_{it})^\prime]$. Define a constant $C^* = (64 C_U C_\Omega q C_d^2+1) C_1$. We then introduce a sparse eigenvalue condition that is essentially the same as the sparse eigenvalue condition considered in much of the Lasso literature \citep{chernozhukov2022automatic,bickel2008simultaneous,belloni2010least,rudelson2012reconstruction}. It will be used to put a bound on the estimation error of $\widehat{\rho}_k$, hence a bound on the error of the estimator $\widehat{\alpha}_k$. 
\begin{assumption}
  \label{assei}
  $\lambda_{max}(B) < C_B < \infty$ holds uniformly in $N$. $\lambda_{min}(G^\prime \Omega G) > C_L > 0$ and $\lambda_{max}(G^\prime \Omega G) < C_U < \infty$ hold uniformly in $N$. Moreover, there exists $C_E>0$ such that with probability approaching one, for $\rho\neq 0$ and $k=1,\ldots,K$, 
  \[
  \min_{J:|J|\leq 2C^*\varepsilon_N^{-2/(2\xi+1)}}\min_{\rho:\|\rho_{J^c}\|_1 \leq 3\|\rho_J\|_1} \frac{\rho^\prime \widehat{G}_k^\prime \widehat{\Omega} \widehat{G}_k\rho}{\rho_J^\prime \rho^{}_J} \geq C_E.
  \]
\end{assumption}

Under these assumptions, we can derive the mean square convergence rate of $\widehat{\alpha}_k(Z_{it})$ as stated in the following theorem, where $F_0$ represents the true distribution of the data. 
\begin{theorem}
  \label{th1}
  If Assumptions \ref{asscp} to \ref{assei} are satisfied and $\varepsilon_N = o(r)$, then for all $k$, 
  \[
  \int [\widehat{\alpha}_k(z_{it}) - \alpha_0(z_{it})]^2 F_0(\rd z_{it}) = O_p(r^2 \varepsilon_N^{-2/(2\xi+1)}).
  \]
\end{theorem}
\citet{chernozhukov2022automatic} develop an analogous result under the cross-sectional setting with exogenous regressors, and \citet{bakhitov2022automatic} develops the convergence rate under the cross-sectional setting in the presence of endogeneity. This theorem extends their results to the panel data setting with two-way fixed effects and endogeneity. It should be noted that we employ conditions that differ from those in \citet{bakhitov2022automatic} yet remain mild, thereby avoiding the issue of a relatively slow convergence rate in a wide range of situations and obtaining a rate identical to that under the cross-sectional setting with exogenous regressors.

Subsequently, we turn to the asymptotic theory of the debiased estimator $\widehat{\theta}_{0t}^d(s)$. We give conditions under which $\sqrt{N}$-consistency and asymptotic normality of $\widehat{\theta}_{0t}^d(s)$ can be established, thereby enabling us to conduct statistical inference in a standard way. 

We first introduce the following boundedness assumption. 
\begin{assumption}
  \label{assa}
  $|\alpha_0(Z_{it})| \leq C_\alpha < \infty$ and $\E[(\Delta \varepsilon_{it})^2 \mid Z_{it}] \leq C_z < \infty$ almost surely. In addition, $\displaystyle \E\Big[\frac{\partial \gamma_0(V_{it})}{\partial D_{i,t-s}}\Big]^2 < \infty$. 
\end{assumption}
\noindent This is a purely technical assumption that ensures certain mean-squared remainders vanish and the variance of the orthogonal moment condition exists so that we may apply the CLT. 

Next, we impose some mild mean square consistency conditions for the function estimators $\widehat{\gamma}_k$ and $\widehat{\alpha}_k$. 
\begin{assumption}
  \label{assm}
  For all $k$, $\displaystyle \int \bigg[\frac{\partial \widehat{\gamma}_k(v_{it})}{\partial D_{i,t-s}} - \frac{\partial \gamma_0(v_{it})}{\partial D_{i,t-s}}\bigg]^2 F_0(\rd v_{it}) \pto 0$, \\ $\int [\Delta \widehat{\gamma}_k(v_{it}) - \Delta \gamma_0(v_{it})]^2 F_0(\rd \overline{v}_{i,t-1:t}) \pto 0$, and $\int [\widehat{\alpha}_k(z_{it}) - \alpha_0(z_{it})]^2 F_0(\rd z_{it}) \pto 0$.
\end{assumption}
\noindent First, this assumption requires that when $\widehat{\gamma}_k$ is plugged into the original moment condition, it will converge in mean square to its true value. Moreover, the differenced estimator $\Delta \widehat{\gamma}_k$ is mean square consistent. Despite the existence of endogeneity, these conditions are easily satisfied, as they allow for the ML estimator to converge at a very slow rate. By Theorem \ref{th1}, we only require $r^2 \varepsilon_N^{-2/(2\xi+1)} \to 0$ for the consistency of $\widehat{\alpha}_k$. 

\begin{assumption}
  \label{asscr}
  Let $F_{0,V_{i,t-1:t} \mid z_{it}}$ denote the conditional distribution of $V_{i,t-1:t}$ given $Z_{it} = z_{it}$, for all $k$,
  \[
  \begin{aligned}
    & \sqrt{N} \bigg\{\int [\widehat{\alpha}_k(z_{it}) - \alpha_0(z_{it})]^2 F_0(\rd z_{it})\bigg\}^{1/2} \\
    & \hspace{2cm} \bigg\{\int \Big\{\int [\Delta \widehat{\gamma}_k(v_{it}) - \Delta \gamma_0(v_{it})]F_{0,V_{i,t-1:t} \mid z_{it}}(\rd v_{i,t-1:t}) \Big\}^2 F_0(\rd z_{it})\bigg\}^{1/2} \pto 0.
  \end{aligned}
  \]
\end{assumption}
\noindent The above expression is equivalent to $\sqrt{N}\|\widehat{\alpha}_k(Z_{it})-\alpha_0(Z_{it})\|\|\E[\Delta \widehat{\gamma}_k(V_{it}) - \Delta \gamma_0(V_{it}) \mid Z_{it}]\| \allowbreak \pto 0$, which requires that the product of the convergence rates of $\|\widehat{\alpha}_k(Z_{it})-\alpha_0(Z_{it})\|$ and $\|\E[\Delta \widehat{\gamma}_k(V_{it}) - \Delta \gamma_0(V_{it}) \mid Z_{it}]\|$ should be faster than $\sqrt{N}$. For $\widehat{\alpha}_k(Z_{it})$, we use the mean square norm and its fast convergence rate is easy to achieve. However, for $\Delta \widehat{\gamma}_k$, we use the projected norm (i.e., the norm of the conditional expectation) rather than the mean square norm. This is because, as mentioned in section \ref{sec31}, the ill-posedness problem will result in very slow convergence rate of the mean square norm. However, it is easy to obtain fast convergence rate of the projected norm \citep[e.g.,][]{blundell2007semi,chen2012estimation,dikkala2020minimax}, as it avoids the ill-posedness. Hence, generally only assumptions like Assumption \ref{asscr} regarding the convergence rate of the projected norm can be satisfied in practice. 
\begin{theorem}
  \label{th2}
  If Assumptions \ref{asscp} to \ref{asscr} are satisfied and $\varepsilon_N = o(r)$, 
  \[
  \sqrt{N} \Big(\widehat{\theta}_{0t}^d(s) - \theta_{0t}(s)\Big) \dto N(0,\Psi),
  \]
  where $\displaystyle \Psi = \E\Big(\frac{\partial \gamma_0(V_{it})}{\partial D_{i,t-s}} - \theta_{0t}(s) + \big[\alpha_0(Z_{it}) - \E[\alpha_0(Z_{it})]\big]\Delta \varepsilon_{it}\Big)^2$. 
\end{theorem}

We also provide a consistent estimator for the asymptotic variance.
\begin{theorem}
  \label{prop1}
  If Assumptions \ref{asscp} to \ref{asscr} are satisfied and $\varepsilon_N = o(r)$, 
  \[
  \widehat{\Psi} = \frac{1}{N} \sum_{k=1}^K \sum_{i\in\mathcal{F}_k} \bigg\{\frac{\partial \widehat{\gamma}_k(V_{it})}{\partial D_{i,t-s}} - \widehat{\theta}_{0t}^d(s) + \big[\widehat{\alpha}_k(Z_{it}) - \frac{1}{n_k}\sum_{j\in\mathcal{F}_k}\widehat{\alpha}_k(Z_{jt})\big][\Delta Y_{it}^* - \Delta \widehat{\gamma}_k^*(V_{it})]\bigg\}^2
  \]
  satisfies $\widehat{\Psi} \pto \Psi$.
\end{theorem}

The asymptotic theory for estimators $\widehat{\theta}_{0t}$ and $\widehat{\theta}_0(0)$ is discussed in Appendix \ref{appa}.

\section{Simulations}
\label{sec5}
We conduct a Monte Carlo study to demonstrate the performance of our debiased estimator. 

First, we introduce the data generating process. We consider three choices for the number of individuals $N=1000,2500,5000$ and three choices for the number of time-varying covariates $L=2,5,10$. The number of time periods is $T=10$. The structural function $\gamma_0(V_{it})$ includes nonlinearities, interactions, and correlations of treatment effects over time. Specifically, the outcome is generated by 
\[
Y_{it} = D_{it} + D_{it}^2 + D_{it}X_{it1} + 0.5D_{i,t-1} + 0.5D_{i,t-1}^2 + 0.5 D_{i,t-1}X_{i,t-1,1} + X_{it}^\prime \xi + \alpha_i + \lambda_t + \varepsilon_{it},
\]
where $\xi = (1^{-2},2^{-2},\ldots,L^{-2})^\prime$, $\{\alpha_i\}_{i=1}^N \stackrel{i.i.d.}{\sim} N(0,1)$, $\{\lambda_t\}_{t=1}^T \stackrel{i.i.d.}{\sim} N(0,1)$, $\{\varepsilon_{it}\}_{i=1,t=1}^{N,T} \stackrel{i.i.d.}{\sim} N(0,1)$. $\{\varepsilon^X_{itl}\}_{i=1,t=1,l=1}^{N,T,L} \stackrel{i.i.d.}{\sim} N(\alpha_i,1)$, $X_{itl}=\varepsilon^X_{itl} + 0.5\varepsilon_{i,t-1}$, $\forall i,t,l$. $\{\varepsilon_{it}^D\}_{i=1,t=1}^{N,T} \stackrel{i.i.d.}{\sim} N(1,1)$, $D_{it} = X_{it}^\prime \xi + \varepsilon_{it}^D$, $\forall i,t$, and $\varepsilon_{i0} = 0$.

The three parameters of interest are $\theta_{0T}(0)$, $\theta_{0T}(1)$, and $\theta_{0T} = 0.5\theta_{0T}(0)+0.5\theta_{0T}(1)$, whose true value is 3, 1.5, and 2.25, respectively.

To save space, this section focuses on the estimation of $\theta_{0T}(0)$ with $N=2500$ and $L=5$ under the above data generating process. Simulation results for other estimators, alternative specifications, and a different data generating process are provided in Appendix \ref{appc}. To implement the DML estimator $\widehat{\theta}_{0T}^d(0)$, we use penalized GMM to estimate $\gamma_0(V_{iT})$, where the regressors include $\{D_{iT}^j,D_{i,T-1}^j,X_{iT}^j\}_{j=1}^3$ and $\{D_{iT}^j X_{iT}^k, D_{i,T-1}^j X_{i,T-1}^k\}_{j=1,k=1}^{2,2}$, and the instruments include $\{D_{i,T-1}^j,D_{i,T-2}^j,X_{i,T-1}^j,X_{i,T-2}^j\}_{j=1}^3$ and $\{D_{i,T-1}^j X_{i,T-1}^k, \\D_{i,T-2}^j X_{i,T-2}^k\}_{j=1,k=1}^{2,2}$. To estimate $\alpha_0(Z_{iT})$, we set the dictionary $b(Z_{iT})$ to be $\{D_{i,T-1}^j, \\ D_{i,T-2}^j,X_{i,T-1}^j,X_{i,T-2}^j, D_{i,T-1}^j X_{i,T-1}^k, D_{i,T-2}^j X_{i,T-2}^k\}_{j=1,k=1}^{3,3}$ with the intercept, and set the dictionary $d(V_{iT})$ to be polynomials with full interactions of $\{D_{iT},D_{i,T-1},X_{iT},X_{i,T-1}\}$ of degree no greater than three. The weighting matrices in the above penalized GMM are set to the identity matrix. For the cross-fitting, we set the number of folds to $K=5$. Moreover, all the hyperparameters are tuned by cross validation. 

We compare our debiased estimator $\widehat{\theta}_{0t}^d(s)$ using penalized GMM for the estimation of $\gamma_0(V_{iT})$, which we call DPGMM, with the plug-in estimator $\widehat{\theta}_{0t}^p(s)$ using penalized GMM for the estimation of $\gamma_0(V_{iT})$ as well, which we call PGMM. To demonstrate the robustness of the proposed algorithm, we make two more comparisons. First, using GMM with untransformed variables without higher-order polynomials and interactions for the estimation of $\gamma_0(V_{iT})$ while keeping everything else unchanged in Algorithm \ref{ag1}, we compare the debiased estimator, which we call DGMM, with the plug-in estimator, which we call GMM. This is to test the debiasing ability of the debiased estimator under severely misspecified first-stage model (since higher-order terms are dropped). Second, using Lasso with transformed variables for the estimation of $\gamma_0(V_{iT})$ while keeping everything else unchanged in Algorithm \ref{ag1}, we compare the debiased estimator, which we call DLasso, with the plug-in estimator, which we call Lasso. This is to test the debiasing ability of the debiased estimator when endogeneity is wrongly ignored (since Lasso assumes that the errors are exogenous). 

The simulation results on 500 simulated samples are reported in Table \ref{tb1}. We can observe that our estimator (DPGMM) has the best performance. Compared with the non-debiased estimator PGMM, our estimator removes nearly all of the bias and has lower standard deviation, leading to lower MSE and significantly improved coverage of confidence intervals. This effectively demonstrates the debiasing and variance reducing power of our algorithm. Moreover, the results of GMM and Lasso indicate that when the model is misspecified or when the endogeneity is ignored, the estimates would be severely biased. However, the debiasing term in DGMM and DLasso can still markedly reduce the bias of the final estimate in these cases. In addition, our variance estimator can provide a relatively accurate estimate of the true variance as well. It can be seen that when not debiased, the variance estimator will significantly underestimate the variance, which results in poor coverage that is far from the nominal level. 

\begin{table}[ht]
  \setlength{\tabcolsep}{12pt}
  \centering
  \caption{Comparison of the performance of DML and other estimators}
  \label{tb1}
  \begin{threeparttable}
  \begin{tabular}{ccccccc}
    \toprule
    & DPGMM & PGMM & DGMM & GMM & DLasso & Lasso \\
    \midrule
    True Value & 3 & 3 & 3 & 3 & 3 & 3 \\
    Estimate & 2.9941 & 2.9226 & 3.1446 & 3.7623 & 3.0081 & 2.8974 \\
    Bias & 0.0059 & 0.0774 & 0.1446 & 0.7623 & 0.0081 & 0.1026 \\
    Std. Dev. & 0.0704 & 0.0832 & 0.0991 & 0.1044 & 0.0738 & 0.0614 \\
    MSE & 0.0050 & 0.0129 & 0.0307 & 0.5920 & 0.0055 & 0.0143 \\
    Est. S.D. & 0.0721 & 0.0570 & 0.0806 & 0.0366 & 0.0671 & 0.0534 \\
    Coverage & 0.948 & 0.66 & 0.53 & 0 & 0.922 & 0.522 \\
    \bottomrule
  \end{tabular}
  \begin{tablenotes}[flushleft]
    \footnotesize
    \item Note: The results are based on 500 simulated samples. ``Estimate'' represents the average of 500 estimates, ``Bias'' is the absolute difference between the average of the 500 estimates and the true value. ``Std. Dev.'' is the standard deviation of 500 estimates, ``MSE'' is the average of the squared errors of 500 estimates. ``Est. S.D.'' is the average of the 500 estimates of the asymptotic standard deviation. For DPGMM, the asymptotic variance estimator is given in Theorem \ref{prop1}. For DGMM and DLasso, asymptotic variance estimators can be similarly constructed. For PGMM, GMM, and Lasso, asymptotic variance estimators are all $\sum_{i=1}^N [\partial \widehat{\gamma}_k(V_{it})/\partial D_{it} - \widehat{\theta}_{0T}^p(0)]^2/N$. ``Coverage'' is the proportion of 500 95\% confidence intervals (based on the asymptotic variance estimator) that contain the true value. 
  \end{tablenotes}
  \end{threeparttable}    
\end{table}

We also plot the distributions of the estimation errors in Figure \ref{fg1}. Compared with the undebiased methods, the distributions of the estimation errors of the debiased methods are closer to 0 and also narrower. Furthermore, they are very close to the normal distribution, which provides evidence for the correctness of the established asymptotic theory.

\begin{figure}[ht]
  \includegraphics[width=\textwidth]{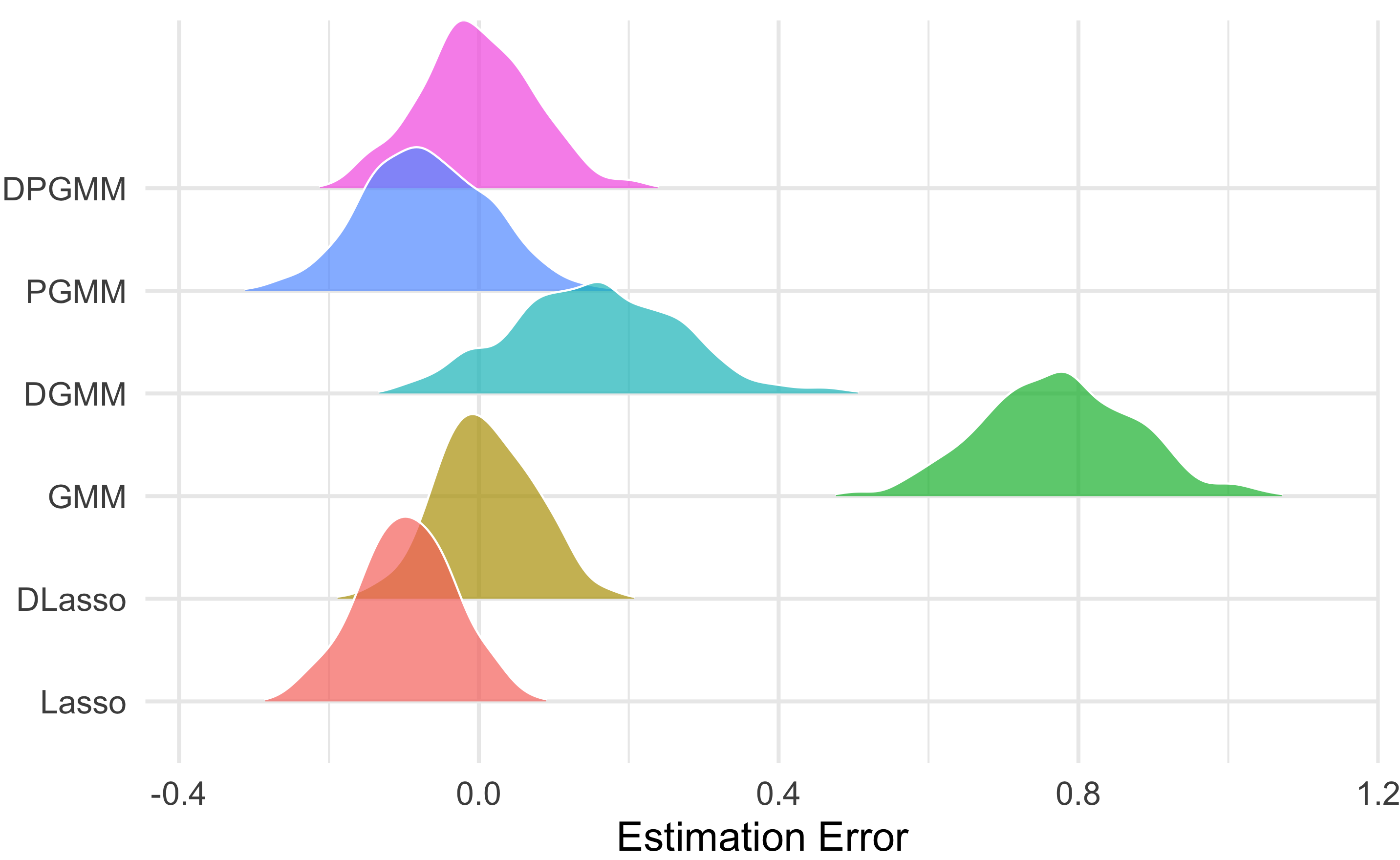}
  \caption{Distributions of the estimation errors of the estimators}
  \label{fg1}  
\end{figure}

In summary, the simulation results indicate that our estimator avoids the model misspecification and high variance issues of GMM, as well as the problem of neglecting endogeneity of Lasso. Meanwhile, it successfully addresses the adverse effects caused by regularization in penalized GMM through the debiasing term.

\section{Empirical application}
\label{sec6}
\subsection{Background}
Understanding the evolution of BMI during early childhood is of great significance for the study of early childhood human capital. Recently, \citet{millimet2015persistence} investigate the persistence and mobility of BMI among American children and find that the mobility of BMI shows a U-shaped relationship with age and the persistence of BMI significantly increases after children enter primary school. \citet{geserick2018acceleration} study the acceleration of BMI in early childhood and the risk of sustained obesity. They find that among obese adolescents, the most rapid weight gain had occurred between 2 and 6 years of age; most children who were obese at that age were obese in adolescence. In this application, we provide further details on this topic. We use the proposed framework to estimate the effects of factors that may influence early childhood BMI so that we can identify what factors influence early childhood BMI and measure the strength of their influence. 

The analysis is based on the Early Childhood Longitudinal Survey, Kindergarten Cohort dataset (ECLS-K), which is also used in \citet{millimet2017dynamic}. It is a panel dataset including data on over 20,000 American students from roughly 1,000 schools who entered kindergarten during the 1998–99 school year. It collects information on BMI, family socioeconomic status, household size, fast food price index, the number of children’s books, the presence of the mother in the household, whether the student exercises at least three days a week, whether the student usually eats a school-provided lunch, and health insurance status for students in fall and spring kindergarten, fall and spring first grade, spring third grade, spring fifth grade and spring eighth grade. Due to data missingness, following \citet{millimet2017dynamic}, we drop the fall first-grade wave and use a balanced subset containing $N=9155$ individuals and $T=6$ periods.\footnote{The dataset is publicly available at \url{https://doi.org/10.15456/jae.2022326.0703753341}.} 

\subsection{Empirical results}
We want to estimate the effect of three candidate treatment variables -- family socioeconomic status (SES), household size (HHS), fast food price index (FFP) -- on the outcome variable, BMI. All variables other than the treatment and outcome will be treated as covariates. First, we construct three nonparametric static panel models with TWFE, where each of the above three variables serves as the treatment in turn. However, \citet{millimet2017dynamic} point out that in studies examining the evolution of early childhood human capital, dynamic panel models are quite commonplace, and existing literature regards BMI as an outcome jointly influenced by both concurrent and historical health-related inputs. Accordingly, we also estimate three nonparametric dynamic panel models with TWFE, again treating each of the three variables as the treatment in turn. To estimate the structural function, we prepare dictionaries containing covariates, cubic polynomials of the current and first-order lagged treatments, and interactions between treatments and covariates. For the dynamic panel models, the dictionaries additionally include the first-order lagged outcome variable. We estimate the average derivatives of the outcome with respect to both the concurrent and first-order lagged treatment across multiple time periods, so we can investigate the lagged effects of the treatment and the temporal trends of treatment effects. 

The estimation results are reported in Table \ref{tb2}, from which several key insights emerge. First, the empirical results from the static and dynamic panel models are broadly comparable in both the direction and statistical significance, with the dynamic one producing weaker magnitudes on some significant effects. Second, regardless of the model specification, the contemporaneous effect of family socioeconomic status on BMI exhibits an inverted U-shaped pattern over time. Specifically, this effect is significantly negative in period 3, becomes insignificant in period 4, turns significantly positive and peaks in period 5, and then declines in period 6 while remaining significantly positive. The one-period lagged effect of family socioeconomic status shows a steady increase from period 4 to period 6, shifting from significantly negative to significantly positive. Turning to household size, under both model specifications, its contemporaneous effect is only significant in periods 3 and 4 -- strongly positive in period 3 but weakly negative in period 4. Its one-period lagged effect is insignificant across all periods. Lastly, under the static panel model, the contemporaneous effect of fast food price is significantly positive only in period 3, while its one-period lagged effect is significantly positive only in period 4. In contrast, both the contemporaneous and one-period lagged effects of fast food price become insignificant when the dynamic panel model is applied.

\begin{table}[ht]
  \setlength{\tabcolsep}{4pt}
  \centering
  \caption{Estimates and standard errors of various average derivatives}
  \label{tb2}
  \begin{threeparttable}
  \begin{tabular}{ccccccccc}
    \toprule
    Model & Treatment & $\theta_{06}(0)$ & $\theta_{06}(1)$ & $\theta_{05}(0)$ & $\theta_{05}(1)$ & $\theta_{04}(0)$ & $\theta_{04}(1)$ & $\theta_{03}(0)$ \\
    \midrule

    \multirow{6}{*}{Static}
    & \multirow{2}{*}{SES}
    & 0.1776 & 0.1817 & 0.4607 & 0.0006 & 0.0289 & -0.4319 & -0.1714 \\
    & & (0.0487) & (0.0370) & (0.0535) & (0.0425) & (0.0911) & (0.0884) & (0.0348) \\

    & \multirow{2}{*}{HHS}
    & 0.0097 & -0.0041 & 0.0018 & -0.0399 & -0.0574 & -0.0081 & 0.1482 \\
    & & (0.0139) & (0.0126) & (0.0132) & (0.0383) & (0.0298) & (0.0483) & (0.0459) \\

    & \multirow{2}{*}{FFP}
    & -0.0134 & 0.1937 & -0.0019 & 0.2349 & -0.0565 & 0.9174 & 0.5806 \\
    & & (0.0971) & (0.1460) & (0.2511) & (0.6352) & (0.5490) & (0.4032) & (0.1403) \\
    \hline

    \multirow{6}{*}{Dynamic}
    & \multirow{2}{*}{SES}
    & 0.1829 & 0.1261 & 0.5070 & -0.0102 & -0.0811 & -0.3668 & -0.1702 \\
    & & (0.0490) & (0.0349) & (0.0525) & (0.0319) & (0.0546) & (0.0252) & (0.0238) \\

    & \multirow{2}{*}{HHS}
    & 0.0093 & -0.0039 & 0.0024 & -0.0365 & -0.0539 & 0.0254 & 0.1198 \\
    & & (0.0140) & (0.0119) & (0.0133) & (0.0325) & (0.0295) & (0.0153) & (0.0301) \\

    & \multirow{2}{*}{FFP}
    & -0.0328 & 0.1865 & 0.0389 & 0.0634 & -0.1332 & 0.4669 & 0.5811 \\
    & & (0.0972) & (0.1412) & (0.2475) & (0.6099) & (0.4897) & (0.4916) & (0.4398) \\
    \bottomrule
  \end{tabular}

  \begin{tablenotes}[flushleft]
    \footnotesize
    \item Note: Numbers not in parentheses are the estimates, and numbers in parentheses are the standard errors. $\theta_{03}(1)$, $\theta_{02}(0)$, $\theta_{02}(1)$, $\theta_{01}(0)$, $\theta_{01}(1)$ are not estimated because data before period 1 are unavailable, so some variables necessary in their estimation cannot be calculated.
  \end{tablenotes}
  \end{threeparttable}
\end{table}

Taken together, these findings indicate that family socioeconomic status exerts the strong-est influence on BMI, followed by household size, while the effect of fast food price is the weakest. Moreover, the attenuation of some coefficients in the dynamic panel model indicates that these effects may be partially explained by past outcomes or omitted time-series components, underscoring the appropriateness of accounting for dynamics in the estimation. Nonetheless, the similarity between the two models' results reinforces the robustness of the findings. It should also be noted that the mechanisms underlying these relationships fall beyond the scope of this study and are therefore not explored further. Nevertheless, our finding that the effect of family socioeconomic status reverses direction as children age may partially explain the puzzle mentioned by \citet{millimet2017dynamic}: some studies report a significant positive effect of family socioeconomic status on BMI, whereas others report a significant negative effect. We argue that these seemingly contradictory findings stem from the fact that existing studies typically estimate a single coefficient for the variable, thereby capturing only the time-aggregated effect of this variable. Because different estimation methods apply different implicit weights across time periods, opposing results are a natural consequence. This highlights a key advantage of our framework. By allowing the estimation of period-specific effects rather than merely the aggregated one, our framework uncovers the temporal dynamics of the effect and sidesteps the long-standing debate over its direction. Finally, the relatively weak or insignificant effects of household size and fast food price are consistent with the conclusions in the existing literature.

\section{Conclusion}
\label{sec7}
This paper develops a double/debiased machine learning framework for estimating average derivative effects in nonparametric panel data models with two-way fixed effects and endogenous regressors. The framework accommodates static and dynamic specifications, sequentially exogenous and completely endogenous covariates, lagged outcomes, and continuous time-varying treatments. To handle the correlation introduced by cross-sectional demeaning, we propose a cross-fitting scheme that restores the independence required by the asymptotic theory; to handle the lack of a closed-form debiasing term under endogeneity, we estimate the Riesz representer via penalized GMM. We establish the convergence rate of this estimator, prove $\sqrt{N}$-consistency and asymptotic normality of the debiased estimators of contemporaneous, dynamic, and aggregated effects, and provide a consistent variance estimator. Monte Carlo evidence and an application to the ECLS-K data confirm the framework's effectiveness.

Two extensions warrant further investigation. First, our theory is developed under fixed $T$ and $N \to \infty$; extending it to a joint $(N, T) \to \infty$ regime would require addressing the incidental parameters problem for time fixed effects. Second, in settings where credible instruments are unavailable, developing a counterpart of our framework under alternative identifying restrictions---such as parallel trends or proxy variables---would broaden the reach of automatic debiasing in panel applications.

\begin{appendices}
\renewcommand{\theequation}{S.\arabic{equation}}
\renewcommand{\thetheorem}{S.\arabic{theorem}}
\renewcommand{\theassumption}{S.\arabic{assumption}}
\renewcommand{\thecorollary}{S.\arabic{corollary}}
\renewcommand{\thetable}{S.\arabic{table}}
\renewcommand{\thefigure}{S.\arabic{figure}}
\renewcommand{\theHequation}{S.\arabic{equation}}
\renewcommand{\theHtheorem}{S.\arabic{theorem}}
\renewcommand{\theHassumption}{S.\arabic{assumption}}
\renewcommand{\theHcorollary}{S.\arabic{corollary}}
\renewcommand{\theHtable}{S.\arabic{table}}
\renewcommand{\theHfigure}{S.\arabic{figure}}
\setcounter{equation}{0}
\setcounter{theorem}{0}
\setcounter{assumption}{0}
\setcounter{corollary}{0}
\setcounter{table}{0}
\setcounter{figure}{0}

\section{Estimators for \texorpdfstring{$\theta_{0t}$}{theta} and \texorpdfstring{$\theta_0(0)$}{theta}}
\label{appa}
In section \ref{sec3} and \ref{sec4} of the main text, we present the estimation and asymptotic theory for $\widehat{\theta}_{0t}^d(s)$ under given $t$ and $s$. For notational simplicity, we omit subscripts $s$ and $t$ for the functions in the main text. In this appendix, we no longer omit the subscripts. First, the estimand is $\theta_{0t}(s) = \E\left[m_s(W_{it}^*,\Delta \gamma_0^*)\right]$, where $\displaystyle m_s(W_{it}^*,\Delta \gamma_0^*) = \frac{\partial \Delta \gamma_0^*(\overline{X}_{i,t-q:t},\overline{D}_{i,t-q:t},C_i)}{\partial D_{i,t-s}}$. The function $\alpha_0(Z_{it})$ is now denoted by $\alpha_{0ts}(Z_{it})$ and will satisfy 
\[
  \E[m_s(W_{it}^*,\Delta \gamma^*)] = \E[\alpha_{0ts}(Z_{it})\Delta \gamma^*(V_{it})] \text{ for all } \gamma \text{ with } \E[\Delta \gamma^*(V_{it})]^2 < \infty.
\]
The orthogonal moment condition becomes $\E[\psi_s(W_{it}^*,\theta_{0t}(s),\Delta \gamma^*_0,\alpha_{0ts})] = 0$, where \\ $\psi_s(W_{it}^*,\theta,\Delta \gamma^*,\alpha) = m_s(W_{it}^*,\Delta \gamma^*) - \theta + \alpha(Z_{it})[\Delta Y_{it}^* - \Delta \gamma^*(V_{it})]$. In Algorithm \ref{ag1}, $\widehat{G}_k$ is now denoted as $\widehat{G}_{kt}$ and we define 
\[
\widehat{M}_{kts} = \frac{1}{N-n_k-n_{k^\prime}}\sum_{i\in\mathcal{F}_{k,k^\prime}^c} (m_s(W_{it}^*,\Delta d_1^*),\ldots,m_s(W_{it}^*,\Delta d_q^*))^\prime
\]
and $\widehat{\rho}_{kts} = \mathop{\arg\min}_{\rho\in\mathbb{R}^p} \big(\widehat{M}_{kts} - \widehat{G}_{kt}\rho\big)^\prime \widehat{\Omega} \big(\widehat{M}_{kts} - \widehat{G}_{kt}\rho\big) + 2 r \|\rho\|_1$. The estimator for $\alpha_{0ts}(Z_{it})$ is $\widehat{\alpha}_{kts} = b(Z_{it})^\prime \widehat{\rho}_{kts}$ and finally the estimator for $\theta_{0t}(s)$ is 
\[
  \widehat{\theta}_{0t}^d(s) = \frac{1}{N} \sum_{k=1}^K \sum_{i\in\mathcal{F}_k} m_s(W_{it}^*,\Delta \widehat{\gamma}_k^*) + \widehat{\alpha}_{kts}(Z_{it})[\Delta Y_{it}^* - \Delta \widehat{\gamma}_k^*(V_{it})].
\]

We first consider the estimator for $\theta_{0t}$, which can be easily constructed by plugging $\widehat{\theta}_{0t}^d(s)$'s into the definition of $\theta_{0t}$. Specifically, we run Algorithm \ref{ag1} for $s=0,\ldots,q$ to obtain $\widehat{\theta}_{0t}^d(0), \ldots, \widehat{\theta}_{0t}^d(q)$, and then build the estimator $\widehat{\theta}_{0t}$ as 
\[
\widehat{\theta}_{0t} = \sum_{s=0}^q w_s \widehat{\theta}_{0t}^d(s).
\]
We have the following result regarding its asymptotic normality. 
\begin{theorem}
  \label{ths1}
  If Assumption \ref{asscp} to \ref{asscr} are satisfied for $s=0,\ldots,q$ and $\varepsilon_N = o(r)$, 
  \[
  \sqrt{N} \Big(\widehat{\theta}_{0t} - \theta_{0t}\Big) \dto N(0,w^\prime \Sigma w),
  \]
  where $w=(w_0,\ldots,w_q)^\prime$, and 
  \[
  \Sigma = \Var\begin{pmatrix}
  \frac{\partial \gamma_0(V_{it})}{\partial D_{i,t}} - \theta_{0t}(0) + \alpha_{0t0}(Z_{it})\Delta \varepsilon_{it} - \E[\alpha_{0t0}(Z_{it})]\Delta \varepsilon_{it} \\
  \cdots \\
  \frac{\partial \gamma_0(V_{it})}{\partial D_{i,t-q}} - \theta_{0t}(q) + \alpha_{0tq}(Z_{it})\Delta \varepsilon_{it} - \E[\alpha_{0tq}(Z_{it})]\Delta \varepsilon_{it} 
\end{pmatrix}.
  \]
\end{theorem}

Then, to construct the estimator for $\theta_0(0)$, we run Algorithm \ref{ag1} for $t=1,\ldots,T$ to obtain $\widehat{\theta}_{01}(0), \ldots, \widehat{\theta}_{0T}(0)$, and then construct the estimator $\widehat{\theta}_0(0)$ as 
\[
\widehat{\theta}_0(0) = \sum_{t=1}^T w_t \widehat{\theta}_{0t}(0).
\]
We similarly have the following asymptotic property. 
\begin{theorem}
  \label{ths2}
  If Assumption \ref{asscp} to \ref{asscr} are satisfied for $t=1,\ldots,T$ and $\varepsilon_N = o(r)$, 
  \[
  \sqrt{N} \Big(\widehat{\theta}_0(0) - \theta_0(0)\Big) \dto N(0,w^\prime \Sigma w),
  \]
  where $w=(w_1,\ldots,w_T)^\prime$, and 
  \[
  \Sigma = \Var\begin{pmatrix}
  \frac{\partial \gamma_0(V_{i1})}{\partial D_{i1}} - \theta_{01}(0) + \alpha_{010}(Z_{i1})\Delta \varepsilon_{i1} - \E[\alpha_{010}(Z_{i1})]\Delta \varepsilon_{i1} \\
  \cdots \\
  \frac{\partial \gamma_0(V_{iT})}{\partial D_{iT}} - \theta_{0T}(0) + \alpha_{0T0}(Z_{iT})\Delta \varepsilon_{iT} - \E[\alpha_{0T0}(Z_{iT})]\Delta \varepsilon_{iT} 
  \end{pmatrix}.
  \]
\end{theorem}

\section{Estimation for the dynamic panel model}
\label{appb}
In this section, we first present the estimation details for the estimand $\theta_{0t}(s)$ under the dynamic panel model \eqref{mod2}.

First, we subtract the model for period $t-s-1$ from the model for period $t$ to obtain  
\begin{equation}
  \label{fdd}
  \Delta Y_{it} = \Delta \gamma_0(\overline{Y}_{i,t-1-p:t-1},\overline{X}_{i,t-q:t},\overline{D}_{i,t-q:t},C_i) + \Delta \lambda_t + \Delta \varepsilon_{it},
\end{equation}
where $\Delta Y_{it} = Y_{it} - Y_{i,t-s-1}$, $\Delta \lambda_t = \lambda_t - \lambda_{t-s-1}$, $\Delta \varepsilon_{it} = \varepsilon_{it} - \varepsilon_{i,t-s-1}$, $\Delta \gamma_0(\overline{Y}_{i,t-1-p:t-1},\overline{X}_{i,t-q:t}, \allowbreak \overline{D}_{i,t-q:t},C_i) = \gamma_0(\overline{Y}_{i,t-1-p:t-1},\overline{X}_{i,t-q:t},\overline{D}_{i,t-q:t},C_i) - \gamma_0(\overline{Y}_{i,t-p-s-2:t-s-2},\overline{X}_{i,t-q-s-1:t-s-1},\\ \overline{D}_{i,t-q-s-1:t-s-1},C_i)$. Then, we subtract the cross-sectional mean of fold $\mathcal{F}_{k^\prime}$ from each term in \eqref{fdd} to eliminate the time fixed effect,
\[
  \Delta Y_{it}^* = \Delta Y_{it} - \frac{1}{|\mathcal{F}_{k^\prime}|}\sum_{j\in\mathcal{F}_{k^\prime}} \Delta Y_{jt} = \Delta \gamma_0^*(\overline{Y}_{i,t-1-p:t-1},\overline{X}_{i,t-q:t},\overline{D}_{i,t-q:t},C_i) + \Delta \varepsilon_{it}^*,
\]
where 
\begin{gather*}
  \begin{aligned}
    \Delta \gamma_0^*(\overline{Y}_{i,t-1-p:t-1},\overline{X}_{i,t-q:t},\overline{D}_{i,t-q:t},C_i) = & \Delta \gamma_0(\overline{Y}_{i,t-1-p:t-1},\overline{X}_{i,t-q:t},\overline{D}_{i,t-q:t}, C_i) \\
    & - \frac{1}{|\mathcal{F}_{k^\prime}|}\sum_{j\in\mathcal{F}_{k^\prime}} \Delta \gamma_0(\overline{Y}_{j,t-1-p:t-1},\overline{X}_{j,t-q:t},\overline{D}_{j,t-q:t},C_j),
  \end{aligned}
  \\
  \text{and }\Delta \varepsilon_{it}^* = \Delta \varepsilon_{it} - \frac{1}{|\mathcal{F}_{k^\prime}|}\sum_{j\in\mathcal{F}_{k^\prime}} \Delta \varepsilon_{jt}.
\end{gather*}
In the dynamic panel scenario, it follows from $\E(\varepsilon_{it} \mid \overline{Y}_{i,t-1},\overline{X}_{1it},\overline{D}_{it},\overline{I}_{it},C_i,\mu_i,\lambda_t) = 0$ that $\E(\Delta \varepsilon_{it} \mid \overline{Y}_{i,t-s-2},\overline{X}_{1,i,t-s-1},\allowbreak \overline{D}_{i,t-s-1},\overline{I}_{i,t-s-1},C_i) = 0$. Therefore, by the i.i.d. assumption, we have 
\begin{equation}
  \E(\Delta \varepsilon_{it}^* \mid \overline{Y}_{i,t-s-2},\overline{X}_{1,i,t-s-1},\overline{D}_{i,t-s-1},\overline{I}_{i,t-s-1},C_i) = 0.
\end{equation}

By the difinition of $\Delta \gamma_0^*$, it is easy to see that 
\begin{align*}
  \theta_{0t}(s) & = \E\left[\frac{\partial \gamma_0(\overline{Y}_{i,t-1-p:t-1},\overline{X}_{i,t-q:t},\overline{D}_{i,t-q:t},C_i)}{\partial D_{i,t-s}}\right] \\
  & = \E\left[\frac{\partial \Delta \gamma_0^*(\overline{Y}_{i,t-1-p:t-1},\overline{X}_{i,t-q:t},\overline{D}_{i,t-q:t},C_i)}{\partial D_{i,t-s}}\right],
\end{align*}
so we may estimate $\theta_{0t}(s)$ based on the RHS. 

Define $\displaystyle m(W_{it}^*,\Delta \gamma_0^*) = \frac{\partial \Delta \gamma_0^*(\overline{Y}_{i,t-1-p:t-1},\overline{X}_{i,t-q:t},\overline{D}_{i,t-q:t},C_i)}{\partial D_{i,t-s}}$, \\$V_{it} = (\overline{Y}_{i,t-1-p:t-1},\overline{X}_{i,t-q:t},\overline{D}_{i,t-q:t},C_i)^\prime$ and $Z_{it} = (\overline{Y}_{i,t-s-2},\overline{X}_{1,i,t-s-1},\overline{D}_{i,t-s-1},\overline{I}_{i,t-s-1},C_i)^\prime$. Then, everything would be the same as the static panel case, and we may leverage Algorithm \ref{ag1} to estimate $\theta_{0t}(s)$. 

After obtaining the estimator $\widehat{\theta}^d_{0t}(s)$, we can easily derive the plug-in estimators for $\theta_{0t}$ and $\theta_0(0)$, just as in the static panel case. For brevity, we omit their explicit presentation here. 

Finally, it is noteworthy that the asymptotic theory developed in section \ref{sec4} and Appendix \ref{appa} applies to the dynamic panel case as well. 

\section{Supplementary simulation results}
\label{appc}
In this appendix, we first present simulation results regarding the estimation of $\theta_{0T}(1)$ and $\theta_{0T}$. The data generating process and implementation details are identical to those in the main text, and the simulation results are reported in Table \ref{tb3} and \ref{tb4}, respectively. 

\begin{table}[ht]
  \setlength{\tabcolsep}{12pt}
  \centering
  \caption{Performance DML and other estimators on the estiamtion of $\theta_{0T}(1)$}
  \label{tb3}
  \begin{threeparttable}
  \begin{tabular}{ccccccc}
    \hline\hline
    & DPGMM & PGMM & DGMM & GMM & DLasso & Lasso \\
    \hline
    True Value & 1.5 & 1.5 & 1.5 & 1.5 & 1.5 & 1.5 \\
    Estimate & 1.4971 & 1.4745 & 1.5306 & 1.8851 & 1.4886 & 1.4154 \\
    Bias & 0.0029 & 0.0255 & 0.0306 & 0.3851 & 0.0114 & 0.0846 \\
    Std. Dev. & 0.0429 & 0.0539 & 0.0693 & 0.0785 & 0.0406 & 0.0372 \\
    MSE & 0.0018 & 0.0036 & 0.0057 & 0.1545 & 0.0018 & 0.0085 \\
    Est. S.D. & 0.0426 & 0.0285 & 0.0656 & 0.0173 & 0.0379 & 0.0259 \\
    Coverage & 0.954 & 0.664 & 0.9 & 0 & 0.928 & 0.182 \\
    \hline
  \end{tabular}
  \begin{tablenotes}[flushleft]
    \footnotesize
    \item Note: Same as those of Table \ref{tb1}.
  \end{tablenotes}
  \end{threeparttable}    
\end{table}

\begin{table}[ht]
  \setlength{\tabcolsep}{12pt}
  \centering
  \caption{Performance DML and other estimators on the estiamtion of $\theta_{0T}$}
  \label{tb4}
  \begin{threeparttable}
  \begin{tabular}{ccccccc}
    \hline\hline
    & DPGMM & PGMM & DGMM & GMM & DLasso & Lasso \\
    \hline
    True Value & 2.25 & 2.25 & 2.25 & 2.25 & 2.25 & 2.25 \\
    Estimate & 2.2477 & 2.1972 & 2.3335 & 2.8286 & 2.2579 & 2.1593 \\
    Bias & 0.0023 & 0.0528 & 0.0835 & 0.5786 & 0.0079 & 0.0907 \\
    Std. Dev. & 0.0414 & 0.0517 & 0.0616 & 0.0628 & 0.0426 & 0.0338\\
    MSE & 0.0017 & 0.0055 & 0.0108 & 0.3387 & 0.0019 & 0.0094 \\
    Est. S.D. & 0.0438 & 0.0344 & 0.0537 & 0.0236 & 0.0390 & 0.0318\\
    Coverage & 0.956 & 0.614 & 0.636 & 0 & 0.952 & 0.208\\
    \hline
  \end{tabular}
  \begin{tablenotes}[flushleft]
    \footnotesize
    \item Note: Same as those of Table \ref{tb1}.
  \end{tablenotes}
  \end{threeparttable}    
\end{table}

It can be seen that consistent with the results for $\theta_{0T}(0)$ presented in the main text, our debiased estimator DPGMM exhibits the smallest bias and MSE for both the one-period lagged effect or the aggregated treatment effect. In fact, irrespective of the first-stage method employed (including GMM where the model is severely misspecified and Lasso where endogeneity is ignored), the corresponding debiased estimators substantially reduce both bias and MSE. Moreover, the conventional variance estimators markedly underestimate the variance when applied to non-debiased estimators, resulting in confidence intervals with very low coverage probabilities. In contrast, the variance estimators associated with the debiased estimators closely match the true variances, and the constructed confidence intervals achieve coverage rates much closer to their nominal levels. Taken together with the results reported in the main text, these findings demonstrate that the proposed framework performs well in estimating various interested treatment effects. 

Next, we examine the performance of the debiased estimator when the sample size and the number of covariates vary, with the experimental results summarized in Table \ref{tb5}. As shown, the debiased estimator achieves notable bias reduction across all combinations of sample sizes and covariate dimensions, yielding substantially smaller MSEs and confidence intervals whose coverage rates are much closer to the nominal levels. As the number of covariates increases, the model complexity also rises; consequently, for a fixed sample size, both the bias and variance of the estimator increase. Conversely, when the number of covariates is held constant, increasing the sample size leads to smaller bias and variance. These results are consistent with theoretical expectations and confirm that our estimator behaves reasonably and performs well across different sample sizes and covariate dimensions. 

\begin{table}[ht]
  \setlength{\tabcolsep}{5pt}
  \centering
  \caption{Comparison of the performance of DML and other estimators}
  \label{tb5}
  \begin{threeparttable}
  \begin{tabular}{ccccccccc}
    \hline\hline
    & \multicolumn{4}{c}{DPGMM} & \multicolumn{4}{c}{PGMM}  \\
    Specification & Bias & S.D. & MSE & Cvg. & Bias & S.D. & MSE & Cvg. \\
    \hline
    $L=2, N=1000$ & 0.0135 & 0.1192 & 0.0144 & 0.946 & 0.0927 & 0.1321 & 0.0260 & 0.72 \\
    $L=2, N=2500$ & 0.0063 & 0.0734 & 0.0054 & 0.952 & 0.0942 & 0.0841 & 0.0159 & 0.604 \\
    $L=2, N=5000$ & 0.0025 & 0.0512 & 0.0026 & 0.948 & 0.0916 & 0.0600 & 0.0120 & 0.416 \\
    $L=5, N=1000$ & 0.0139 & 0.1295 & 0.0170 & 0.94 & 0.0749 & 0.1221 & 0.0205 & 0.786 \\
    $L=5, N=2500$ & 0.0073 & 0.0704 & 0.0050 & 0.948 & 0.0737 & 0.0856 & 0.0128 & 0.668 \\
    $L=5, N=5000$ & 0.0037 & 0.0475 & 0.0023 & 0.954 & 0.0728 & 0.0643 & 0.0094 & 0.53 \\
    $L=10, N=1000$ & 0.0172 & 0.1540 & 0.0240 & 0.91 & 0.0645 & 0.1296 & 0.0209 & 0.788 \\
    $L=10, N=2500$ & 0.0068 & 0.0790 & 0.0063 & 0.932 & 0.0625 & 0.0803 & 0.0104 & 0.712 \\
    $L=10, N=5000$ & 0.0052 & 0.0462 & 0.0022 & 0.958 & 0.0584 & 0.0593 & 0.0069 & 0.626 \\
    \hline
  \end{tabular}
  \end{threeparttable}    
\end{table}

Finally, we compare the performance of the six estimators under a different data generating process. The outcome is generated by 
\begin{align*}
  Y_{it} = & D_{it} + D_{it}^2 + D_{i,t-1} + D_{it}X_{it1} + X_{it1}^2/2 - \cos(X_{it2})/3 + \tanh(X_{it3})/4 \\
  & - \exp(X_{it4})/5 + \alpha_i + \lambda_t + \varepsilon_{it},  
\end{align*}
where $\{\alpha_i\}_{i=1}^N \stackrel{i.i.d.}{\sim} N(0,1)$, $\{\lambda_t\}_{t=1}^T \stackrel{i.i.d.}{\sim} N(0,1)$, $\{\varepsilon_{it}\}_{i=1,t=1}^{N,T} \stackrel{i.i.d.}{\sim} N(0,1)$. $\{\varepsilon^X_{itl}\}_{i=1,t=1,l=1}^{N,T,L} \stackrel{i.i.d.}{\sim} N(\alpha_i,1)$, $X_{itl}=\varepsilon^X_{itl} + 0.5\varepsilon_{i,t-1}$, $\forall i,t,l$. $\{\varepsilon_{it}^D\}_{i=1,t=1}^{N,T} \stackrel{i.i.d.}{\sim} N(1,1)$, and $D_{it} = X_{it1}/2 - X_{it2}/3 + X_{it3}/4 + \varepsilon_{it}^D$, $\forall i,t$. In this model, the confounding effects of the time-variant covariates become non-linear. The construction of each estimator follows the same approach as the simulations in the main text. Table \ref{tb6} reports the estimation results. Overall, under the new data generating process, the proposed DPGMM estimator exhibits the best performance in terms of MSE, confidence interval coverage rate, and estimation accuracy of the asymptotic variance. Across all competing estimators, the DPGMM effectively eliminates the bias inherent in the penalized GMM, while also overcoming the bias induced by model misspecification in GMM and the bias resulting from ignoring endogeneity in Lasso. This demonstrates that the proposed estimator is robust to various settings and can deliver consistently favorable performance. 

\begin{table}[ht]
  \setlength{\tabcolsep}{12pt}
  \centering
  \caption{Comparison of the performance of DML and other estimators}
  \label{tb6}
  \begin{threeparttable}
  \begin{tabular}{ccccccc}
    \hline\hline
    & DPGMM & PGMM & DGMM & GMM & DLasso & Lasso \\
    \hline
    True Value & 3 & 3 & 3 & 3 & 3 & 3 \\
    Estimate & 2.9906 & 2.9232 & 3.1092 & 3.7165 & 2.9893 & 2.7294 \\
    Bias & 0.0094 & 0.0768 & 0.1092 & 0.7165 & 0.0107 & 0.2706 \\
    Std. Dev. & 0.1193 & 0.1252 & 0.1769 & 0.2017 & 0.1444 & 0.2531 \\
    MSE & 0.0143 & 0.0216 & 0.0432 & 0.5541 & 0.0210 & 0.1373 \\
    Est. S.D. & 0.1088 & 0.0748 & 0.1686 & 0.0735 & 0.1032 & 0.0709 \\
    Coverage & 0.956 & 0.71 & 0.894 & 0.006 & 0.94 & 0.232 \\
    \hline
  \end{tabular}
  \begin{tablenotes}[flushleft]
    \footnotesize
    \item Note: Same as the note of Table \ref{tb1}.
  \end{tablenotes}
  \end{threeparttable}    
\end{table}

\section{Proofs}
\subsection{Proof of Theorem \ref{th1}}
Define $\overline{\rho}$ as the population coefficients, i.e.,
\begin{equation}
  \label{eqrho}
  \overline{\rho} \in \mathop{\arg\min}_{v\in\mathbb{R}^p} \big(M-G\rho\big)^\prime \Omega \big(M-G\rho\big),
\end{equation}
where $M = \E[(m(W_{it}^*,\Delta d_1^*),\ldots,m(W_{it}^*,\Delta d_q^*))^\prime]$, and denote $\overline{\alpha}(Z_{it}) = b(Z_{it})^\prime \overline{\rho}$. Note that it follows from \eqref{eqrie} and the boundness of the derivatives in Assumption \ref{assbd} that 
\[
M = \E[\Delta d^*(V_{it})\alpha_0(Z_{it})], 
\]
which by the minimization property of $\overline{\rho}$ suggests that for any $\rho$, 
\begin{equation}
  \label{eqle}
  \begin{aligned}
    & \E[\Delta d^*(V_{it})\alpha_0(Z_{it}) - \Delta d^*(V_{it}) b(Z_{it})\overline{\rho}]^\prime \Omega \E[\Delta d^*(V_{it})\alpha_0(Z_{it}) - \Delta d^*(V_{it}) b(Z_{it})\overline{\rho}] \\
    & \leq \E[\Delta d^*(V_{it})\alpha_0(Z_{it}) - \Delta d^*(V_{it}) b(Z_{it})\rho]^\prime \Omega \E[\Delta d^*(V_{it})\alpha_0(Z_{it}) - \Delta d^*(V_{it}) b(Z_{it})\rho].
  \end{aligned}
\end{equation}

Define $\rho_\star$ as 
\begin{equation}
  \label{eqdrs}
  \rho_\star \in \mathop{\arg\min}_{v \in \mathbb{R}^p} (\overline{\rho}-v)^\prime G^\prime\Omega G (\overline{\rho}-v) + 2\varepsilon_N\sum_{j\in J_0^c} |v_j|.
\end{equation}

Choose an integer $s_0$ satisfying $\varepsilon_N^{-2/(2\xi+1)} \leq s_0 \leq C_1 \varepsilon_N^{-2/(2\xi+1)}$. Applying Assumption \ref{asscs} with $m=s_0$, we can define $J_0$ as the indices of a sparse approximation with $|J_0| = s_0$ ($|J_0|$ is the cardinality of set $J_0$) and coefficients $\{\ell_j\}_{j\in J_0}$ such that for $\widetilde{\alpha}(Z_{it}) = \sum_{j\in J_0}\ell_j b_j(Z_{it})$, 
\begin{equation}
  \label{eqpf1}
  \|\alpha_0(Z_{it})-\widetilde{\alpha}(Z_{it})\|^2 \leq C_1 s_0^{-2\xi}.
\end{equation}

In what follows, we will introduce some important lemmas for the proof of Theorem \ref{th1}. We assume that the assumptions in the main text hold true in these lemmas.
\begin{lemma}
  \label{lemma1}
  $(\overline{\rho} - \rho_\star)^\prime G^\prime \Omega G (\overline{\rho} - \rho_\star) = O(\varepsilon_N^{4\xi/(2\xi+1)})$. 
\end{lemma}
\begin{proof}
  Define $\widetilde{\rho} = (\widehat{\rho}_1,\ldots,\widetilde{\rho}_p)$ with 
  \[
  \widetilde{\rho}_j = \begin{cases}
    \ell_j & \text{ if } j\in J_0, \\
    0 & \text{ otherwise}.
  \end{cases}
  \]
  By the definition of $\rho_\star$ and the fact that $\widetilde{\alpha}(Z_{it}) = b(Z_{it})^\prime \widetilde{\rho}$, we have 
  \begin{align*}
    & (\overline{\rho}-\rho_\star)^\prime G^\prime \Omega G (\overline{\rho}-\rho_\star) + 2\varepsilon_N\sum_{j\in J_0^c} |\rho_{\star j}| \\ 
    & \leq (\overline{\rho}-\widetilde{\rho})^\prime G^\prime \Omega G (\overline{\rho}-\widetilde{\rho}) + 2\varepsilon_N\sum_{j\in J_0^c} |\widetilde{\rho}_j| = (\overline{\rho}-\widetilde{\rho})^\prime G^\prime \Omega G (\overline{\rho}-\widetilde{\rho}) \\
    & = \E\{\Delta d^*[(\alpha_0 - b^\prime \widetilde{\rho}) - (\alpha_0 - b^\prime \overline{\rho})]\}^\prime \Omega \E\{\Delta d^*[(\alpha_0 - b^\prime \widetilde{\rho}) - (\alpha_0 - b^\prime \overline{\rho})]\}\\
    & \leq 2 \E[\Delta d^*(\alpha_0 - b^\prime \widetilde{\rho})]^\prime \Omega \E[\Delta d^*(\alpha_0 - b^\prime \widetilde{\rho})] + 2 \E[\Delta d^*(\alpha_0 - b^\prime \overline{\rho})]^\prime \Omega \E[\Delta d^*(\alpha_0 - b^\prime \overline{\rho})] \\
    & \leq 4 \E[\Delta d^*(\alpha_0 - b^\prime \widetilde{\rho})]^\prime \Omega \E[\Delta d^*(\alpha_0 - b^\prime \widetilde{\rho})] \\
    & \leq 4 \E\{[\Delta d^*(\alpha_0 - b^\prime \widetilde{\rho})]^\prime \Omega [\Delta d^*(\alpha_0 - b^\prime \widetilde{\rho})]\} \\
    & \leq 4 \lambda_{max}(\Omega) \|\Delta d^*(\alpha_0 - b^\prime \widetilde{\rho})\|^2 \\
    & \leq 64 \lambda_{max}(\Omega) q C_d^2 \|\alpha_0 - \widetilde{\alpha}\|^2 \\
    & \leq 64 \lambda_{max}(\Omega) q C_d^2 C_1 s_0^{-2\xi}\\
    & \leq 64 C_\Omega q C_d^2 C_1 \varepsilon_N^{4\xi/(2\xi+1)},
  \end{align*}
  where the fifth line follows from \eqref{eqle}, the eighth line follows from Assumption \ref{assbd} (note that if $|d_j|\leq C_d$, then $|\Delta d_j| < 2C_d$ and $|\Delta d_j^*| \leq 4 C_d$), the ninth line follows from \eqref{eqpf1}, and the last line holds because of Assumption \ref{asscp} and $\varepsilon_N^{-2/(2\xi+1)} \leq s_0$. The result then follows from the fact that $\varepsilon_N\sum_{j\in J_0^c} |\rho_{\star j}| \geq 0$.
\end{proof}

\begin{lemma}
  \label{lemma2}
  $\|G^\prime\Omega G(\overline{\rho}-\rho_\star)\|_\infty = O(\varepsilon_N)$. 
\end{lemma}
\begin{proof}
  Let $e_j\in\mathbb{R}^p$ denote the $j$-th colomn of the identity matrix $I_p$. The FOC for $\rho_\star$ implies that for $j\in J_0$, we have $e_j^\prime G^\prime\Omega G(\rho_\star-\overline{\rho}) = 0$; for $j\in J_0^c$, we have $e_j^\prime G^\prime \Omega G(\rho_\star-\overline{\rho}) + \varepsilon_N z_j = 0$, where $z_j = \mathrm{sign}(\rho_{\star j})$ if $\rho_{\star j}\neq 0$ and $z_j \in [-1,1]$ if $\rho_{\star j}=0$. Therefore, for any $j$, we have $|e_j^\prime G^\prime \Omega G (\rho_\star-\overline{\rho})|\leq \varepsilon_N$. Hence, $\|G^\prime\Omega G(\overline{\rho}-\rho_\star)\|_\infty \leq \varepsilon_N$ and the result follows.
\end{proof}
\begin{lemma}
  \label{lemma3}
  $\|\alpha_0(Z_{it}) - b(Z_{it})^\prime \rho_\star\|^2 = O(\varepsilon_N^{4\xi/(2\xi+1)})$.
\end{lemma}
\begin{proof}
  It follows from Assumption \ref{assei} and the proof of Lemma \ref{lemma1} that 
  \[
  C_L \|\overline{\rho}-\widetilde{\rho}\|_2^2 \leq \lambda_{min}(G^\prime\Omega G) \|\overline{\rho}-\widetilde{\rho}\|_2^2 \leq (\overline{\rho}-\widetilde{\rho})^\prime G^\prime\Omega G (\overline{\rho}-\widetilde{\rho}) = O(\varepsilon_N^{4\xi/(2\xi+1)}),
  \]
  which implies  
  \[
  \|\overline{\rho}-\widetilde{\rho}\|_2^2 = O(\varepsilon_N^{4\xi/(2\xi+1)}).
  \]
  Similarly, it follows from Lemma \ref{lemma1} and Assumption \ref{assei} that 
  \[
  C_L \|\overline{\rho} - \rho_\star\|_2^2 \leq \lambda_{min}(G^\prime\Omega G) \|\overline{\rho} - \rho_\star\|_2^2 \leq (\overline{\rho} - \rho_\star)^\prime G^\prime \Omega G (\overline{\rho} - \rho_\star) = O(\varepsilon_N^{4\xi/(2\xi+1)}),
  \]
  which implies 
  \[
  \|\overline{\rho} - \rho_\star\|_2^2 = O(\varepsilon_N^{4\xi/(2\xi+1)}).
  \]
  Consequently, it follows from \eqref{eqpf1} and the above results that   
  \begin{align*}
    & \|\alpha_0(Z_{it}) - b(Z_{it})^\prime \rho_\star\|^2 \\
    & \leq 3\|\alpha_0(Z_{it}) - b(Z_{it})^\prime \widetilde{\rho}\|^2 + 3\|b(Z_{it})^\prime \widetilde{\rho} - b(Z_{it})^\prime \overline{\rho}\|^2 + 3\|b(Z_{it})^\prime \overline{\rho} - b(Z_{it})^\prime \rho_\star\|^2 \\
    & = 3\|\alpha_0(Z_{it}) - b(Z_{it})^\prime \widetilde{\rho}\|^2 + 3(\widetilde{\rho}-\overline{\rho})^\prime B (\widetilde{\rho}-\overline{\rho}) + 3(\overline{\rho}-\rho_\star)^\prime B (\overline{\rho}-\rho_\star)\\
    & \leq 3C_1 s_0^{-2\xi} + 3 \lambda_{max}(B)\|\widetilde{\rho}-\overline{\rho}\|_2^2 + 3 \lambda_{max}(B)\|\overline{\rho}-\rho_\star\|_2^2 \\
    & \leq 3C_1 \varepsilon_N^{4\xi/(2\xi+1)} + 3C_B O(\varepsilon_N^{4\xi/(2\xi+1)}) + 3C_B O(\varepsilon_N^{4\xi/(2\xi+1)}) \\ 
    & = O(\varepsilon_N^{4\xi/(2\xi+1)}).
  \end{align*}
\end{proof}

\begin{lemma}
  \label{lemma4}
  Let $A$ be a $p \times q$ matrix and $B$ be a $q \times m$ matrix. We have 
  \[
  \|AB\|_\infty \leq q\|A\|_\infty\|B\|_\infty.
  \]
\end{lemma}
\begin{proof}
  Let $a_1,\ldots,a_p$ and $b_1,\ldots,b_m$ denote the row vectors of $A$ and column vectors of $B$, respectively. By the H\"{o}lder's inequality, 
  \[
  \begin{aligned}
    & \|AB\|_\infty = \max_{1\leq j \leq p}\max_{1\leq k \leq m} |a_jb_k| \leq \max_{1\leq j \leq p}\max_{1\leq k \leq m} \|a_j\|_1 \|b_k\|_\infty \\
    = & \max_{1\leq j \leq p}\|a_j\|_1 \max_{1\leq k \leq m}\|b_k\|_\infty = \max_{1\leq j \leq p}\|a_j\|_1 \|B\|_\infty\\
    \leq & \max_{1\leq j \leq p} q \max_{1\leq l\leq q}|a_{jl}| \|B\|_\infty = q\|A\|_\infty\|B\|_\infty.
  \end{aligned}
  \]
\end{proof}

\begin{lemma}
  \label{lemma6}
  $\|\widehat{G}_k-G\|_\infty = O_p(\varepsilon_N)$.
\end{lemma}
\begin{proof}
  By the definition of $\widehat{G}_k$, we have 
  \[
  \widehat{G}_k - G = \frac{1}{N-n_k-n_{k^\prime}}\sum_{k^*\in\{1,\ldots,K\}\backslash\{k,k^\prime\}}\sum_{i\in\mathcal{F}_{k^*}} \Big\{\Delta d^*(V_{it})b(Z_{it})^\prime - \E[\Delta d^*(V_{it})b(Z_{it})^\prime]\Big\}.
  \]
  Consider the sum for fold $k^*$ in the above equation,  
  \begin{align}
    & \frac{1}{n_{k^*}} \sum_{i\in\mathcal{F}_{k^*}} \Big\{\Delta d^*(V_{it})b(Z_{it})^\prime - \E[\Delta d^*(V_{it})b(Z_{it})^\prime]\Big\} \notag\\
    = & \frac{1}{n_{k^*}}\sum_{i\in\mathcal{F}_{k^*}}\bigg\{\Big[\Delta d(V_{it}) - \frac{1}{n_{k^{*\prime}}}\sum_{m\in\mathcal{F}_{k^{*\prime}}}\Delta d(V_{mt})\Big]b(Z_{it})^\prime \notag \\
    & - \Big[\E[\Delta d(V_{it})b(Z_{it})^\prime] - \E[\Delta d(V_{it})]\E[b(Z_{it})^\prime]\Big]\bigg\} \notag \\
    = & \frac{1}{n_{k^*}} \sum_{i\in\mathcal{F}_{k^*}} \Big\{\Delta d(V_{it})b(Z_{it})^\prime - \E[\Delta d(V_{it})b(Z_{it})^\prime] \Big\} \notag \\
    & - \bigg[\frac{1}{n_{k^{*\prime}}}\sum_{m\in\mathcal{F}_{k^{*\prime}}}\Delta d(V_{mt})\frac{1}{n_{k^*}}\sum_{i\in\mathcal{F}_{k^*}}b(Z_{it})^\prime - \E[\Delta d(V_{it})]\E[b(Z_{it})^\prime]\bigg]. \label{eqdg}
  \end{align}

  By Assumption \ref{assbd}, 
  \[
  |\Delta d_j(V_{it})b_l(Z_{it})| \leq 2C_dC_b.
  \]
  By the Hoeffding's inequality, for all $j$ and $l$, 
  \[
  \begin{aligned}
    \rP\bigg(\bigg|\frac{1}{n_{k^*}}\sum_{i\in\mathcal{F}_{k^*}} \Big\{\Delta d_j(V_{it})b_l(Z_{it}) - \E[\Delta d_j(V_{it})b_l(Z_{it})]\Big\} \bigg| \geq t \bigg) & \leq 2\exp\Big(-\frac{2n_{k^*}^2t^2}{16n_{k^*}C_d^2C_b^2}\Big) \\
    & = 2\exp\Big(-\frac{n_{k^*}t^2}{8C_d^2C_b^2}\Big).
  \end{aligned}
  \]
  Therefore, 
  \begin{align*}
    & \rP\bigg(\bigg\|\frac{1}{n_{k^*}}\sum_{i\in\mathcal{F}_{k^*}} \Big\{\Delta d(V_{it})b(Z_{it})^\prime - \E[\Delta d(V_{it})b(Z_{it})^\prime]\Big\}\bigg\|_\infty \geq t\bigg) \\
    = & \rP\left(\max_{j,l} \left|\frac{1}{n_{k^*}}\sum_{i\in\mathcal{F}_{k^*}} \Big\{\Delta d(V_{it})b(Z_{it})^\prime - \E[\Delta d(V_{it})b(Z_{it})^\prime]\Big\} \right| \geq t \right) \\
    \leq & \sum_{j=1}^q \sum_{l=1}^p \rP\left(\left|\frac{1}{n_{k^*}}\sum_{i\in\mathcal{F}_{k^*}} \Big\{\Delta d_j(V_{it})b_l(Z_{it}) - \E[\Delta d_j(V_{it})b_l(Z_{it})]\Big\} \right| \geq t \right) \\
    = & 2q^2\exp\Big(-\frac{n_{k^*}t^2}{8C_d^2C_b^2}\Big).
  \end{align*}
  Let $t = K\sqrt{\ln(q)/n_{k^*}}$, where $K$ is a positive constant to be determined. We have 
  \begin{align*}
    & \rP\bigg(\bigg\|\frac{1}{n_{k^*}}\sum_{i\in\mathcal{F}_{k^*}} \Big\{\Delta d(V_{it})b(Z_{it})^\prime - \E[\Delta d(V_{it})b(Z_{it})^\prime] \Big\} \bigg\|_\infty \geq K\sqrt{\ln(q)/n_{k^*}}\bigg)\\
    \leq & 2q^2\exp\Big(-\frac{\ln(q)K^2}{8C_d^2C_b^2}\Big) = 2q^{2-K^2/(8C_d^2C_b^2)}.
  \end{align*}
  For any $\delta > 0$, we have $2q^{2-K^2/(8C_d^2C_b^2)} < \delta$ when $K$ is large enough. Thus, 
  \[
  \rP\bigg(\bigg\|\frac{1}{n_{k^*}}\sum_{i\in\mathcal{F}_{k^*}} \Big\{\Delta d(V_{it})b(Z_{it})^\prime - \E[\Delta d(V_{it})b(Z_{it})^\prime]\Big\}\bigg\|_\infty\Big/\sqrt{\ln(q)/n_{k^*}} \geq K\bigg) \leq \delta
  \]
  for all $n_{k^*}$, which implies 
  \begin{equation}
    \label{eqpf3}
    \bigg\|\frac{1}{n_{k^*}}\sum_{i\in\mathcal{F}_{k^*}} \Big\{\Delta d(V_{it})b(Z_{it})^\prime - \E[\Delta d(V_{it})b(Z_{it})^\prime]\Big\}\bigg\|_\infty = O_p(\sqrt{\ln(q)/n_{k^*}}) = O_p(\sqrt{\ln(q)/N}).
  \end{equation}

  Next, note that 
  \begin{align*}
    & \bigg\|\frac{1}{n_{k^{*\prime}}}\sum_{m\in\mathcal{F}_{k^{*\prime}}}\Delta d(V_{mt})\frac{1}{n_{k^*}}\sum_{i\in\mathcal{F}_{k^*}}b(Z_{it})^\prime - \E[\Delta d(V_{it})]\E[b(Z_{it})^\prime]\bigg\|_\infty \\
    \leq & \bigg\|\frac{1}{n_{k^{*\prime}}}\sum_{m\in\mathcal{F}_{k^{*\prime}}}\Delta d(V_{mt})\Big[\frac{1}{n_{k^*}}\sum_{i\in\mathcal{F}_{k^*}}b(Z_{it})^\prime-\E[b(Z_{it})^\prime]\Big]\bigg\|_\infty \\
    & + \bigg\|\Big[\frac{1}{n_{k^{*\prime}}}\sum_{m\in\mathcal{F}_{k^{*\prime}}}\Delta d(V_{mt})-\E[\Delta d(V_{it})]\Big]\E[b(Z_{it})^\prime]\bigg\|_\infty.
  \end{align*}
  It follows from Assumption \ref{assbd} and the Hoeffding's inequality, we can similarly prove that 
  \[
    \bigg\|\frac{1}{n_{k^*}}\sum_{i\in\mathcal{F}_{k^*}}b(Z_{it})^\prime-\E[b(Z_{it})^\prime]\bigg\|_\infty = O_p(\sqrt{\ln(p)/N})
  \]
  and 
    \[
    \bigg\|\frac{1}{n_{k^{*\prime}}}\sum_{m\in\mathcal{F}_{k^{*\prime}}}\Delta d(V_{mt})-\E[\Delta d(V_{it})]\bigg\|_\infty = O_p(\sqrt{\ln(q)/N}).
  \]
  Thus, we have by the Assumption \ref{assbd} that 
  \begin{align*}
    & \bigg\|\frac{1}{n_{k^{*\prime}}}\sum_{m\in\mathcal{F}_{k^{*\prime}}}\Delta d(V_{mt})\Big[\frac{1}{n_{k^*}}\sum_{i\in\mathcal{F}_{k^*}}b(Z_{it})^\prime-\E[b(Z_{it})^\prime]\Big]\bigg\|_\infty \\
    = & \bigg\|\frac{1}{n_{k^{*\prime}}}\sum_{m\in\mathcal{F}_{k^{*\prime}}}\Delta d(V_{mt})\bigg\|_\infty \bigg\|\frac{1}{n_{k^*}}\sum_{i\in\mathcal{F}_{k^*}}b(Z_{it})^\prime-\E[b(Z_{it})^\prime]\bigg\|_\infty \\
    = & O_p(1)O_p(\sqrt{\ln(p)/N}) = O_p(\sqrt{\ln(q)/N})
  \end{align*}
  and 
  \begin{align*}
    & \bigg\|\Big[\frac{1}{n_{k^{*\prime}}}\sum_{m\in\mathcal{F}_{k^{*\prime}}}\Delta d(V_{mt})-\E[\Delta d(V_{it})]\Big]\E[b(Z_{it})^\prime]\bigg\|_\infty \\
    = & \bigg\|\frac{1}{n_{k^{*\prime}}}\sum_{m\in\mathcal{F}_{k^{*\prime}}}\Delta d(V_{mt})-\E[\Delta d(V_{it})]\bigg\|_\infty \|\E[b(Z_{it})^\prime]\|_\infty \\
    = & O_p(\sqrt{\ln(q)/N}) O(1) = O_p(\sqrt{\ln(q)/N}).
  \end{align*}
  These would imply that 
  \begin{equation}
    \label{eqpf2}
    \bigg\|\frac{1}{n_{k^{*\prime}}}\sum_{m\in\mathcal{F}_{k^{*\prime}}}\Delta d(V_{mt})\frac{1}{n_{k^*}}\sum_{i\in\mathcal{F}_{k^*}}b(Z_{it})^\prime - \E[\Delta d(V_{it})]\E[b(Z_{it})^\prime]\bigg\|_\infty = O_p(\sqrt{\ln(q)/N}).
  \end{equation}
  
  By the triangle inequality, \eqref{eqdg}, \eqref{eqpf3}, \eqref{eqpf2}, and the definition of $\varepsilon_N$ that $\sqrt{\ln(q)/N} \leq \varepsilon_N$, we have
  \[
  \bigg\|\frac{1}{n_{k^*}} \sum_{i\in\mathcal{F}_{k^*}} \Big\{\Delta d^*(V_{it})b(Z_{it})^\prime - \E[\Delta d^*(V_{it})b(Z_{it})^\prime]\Big\}\bigg\|_\infty = O_p(\varepsilon_N).
  \]
  Finally, it follows from the triangle inequality that 
  \begin{align*}
    & \|\widehat{G}_k - G\|_\infty \\
    \leq & \sum_{k^*\in\{1,\ldots,K\}\backslash\{k,k^\prime\}} \frac{n_{k^*}}{N-n_k-n_{k^\prime}} \bigg\|\frac{1}{n_{k^*}}\sum_{i\in\mathcal{F}_{k^*}} \Big\{\Delta d^*(V_{it})b(Z_{it})^\prime - \E[\Delta d^*(V_{it})b(Z_{it})^\prime]\Big\}\bigg\|_\infty \\
    = & \sum_{k^*\in\{1,\ldots,K\}\backslash\{k,k^\prime\}} \Big(\frac{1}{K-2}+o(1)\Big) O_p(\varepsilon_N) = O_p(\varepsilon_N).
  \end{align*}
\end{proof}

\begin{lemma}
  \label{lemmam}
  $\|\widehat{M}_k-M\|_\infty = O_p(\varepsilon_N)$.
\end{lemma}
\begin{proof}
Note that 
\begin{align*}
  & \|\widehat{M}_k-M\|_\infty \\
  = & \frac{1}{N-n_k-n_{k^\prime}} \sum_{i\in\mathcal{F}_{k,k^\prime}^c} \bigg\{\bigg(\frac{\partial \Delta d_1^*(V_{it})}{\partial D_{i,t-s}},\ldots,\frac{\partial \Delta d_q^*(V_{it})}{\partial D_{i,t-s}}\bigg)^\prime \\
  & \hspace{3cm} -\E\bigg[\bigg(\frac{\partial \Delta d_1^*(V_{it})}{\partial D_{i,t-s}},\ldots,\frac{\partial \Delta d_q^*(V_{it})}{\partial D_{i,t-s}}\bigg)^\prime\bigg]\bigg\} \\
  = & \frac{1}{N-n_k-n_{k^\prime}} \sum_{i\in\mathcal{F}_{k,k^\prime}^c} \bigg\{\bigg(\frac{\partial d_1(V_{it})}{\partial D_{i,t-s}},\ldots,\frac{\partial d_q(V_{it})}{\partial D_{i,t-s}}\bigg)^\prime \\
  & \hspace{3cm} -\E\bigg[\bigg(\frac{\partial d_1(V_{it})}{\partial D_{i,t-s}},\ldots,\frac{\partial d_q(V_{it})}{\partial D_{i,t-s}}\bigg)^\prime\bigg]\bigg\}.
\end{align*}
Note that Assumption \ref{assbd} bounds each derivative in the last line, so by Assumption \ref{assbd} and the Hoeffding's inequality,
\[
\rP\Big(\Big\|\widehat{M}_k-M\Big\|_\infty \geq t\Big) \leq 2q \exp\Big(-\frac{(N-n_k-n_{k^\prime})t^2}{2C_m^2}\Big).
\]
Then, by an argument similar to the derivation of \eqref{eqpf3}, 
\[
\Big\|\widehat{M}_k-M\Big\|_\infty = O_p(\varepsilon_N).
\]
\end{proof}

\begin{lemma}
  \label{lemmagh}
  $\|\widehat{G}_k^\prime\widehat{\Omega}\widehat{M}_k - G^\prime \Omega M\|_\infty = O_p(\varepsilon_N)$. 
\end{lemma}
\begin{proof}
  Note that
\[ 
\begin{aligned}
  \widehat{G}_k^\prime\widehat{\Omega}\widehat{M}_k - G^\prime \Omega M & = \widehat{G}_k^\prime (\widehat{\Omega} - \Omega) \widehat{M}_k + (\widehat{G}_k-G)^\prime \Omega \widehat{M}_k + G^\prime \Omega (\widehat{M}_k - M),
\end{aligned}
\]
and by Assumption \ref{assbd}, $\|G\|_\infty = O(1)$, $\|\widehat{G}_k\|_\infty = O_p(1)$, and $\|\widehat{M}_k\|_\infty = O_p(1)$. Therefore, by Lemma \ref{lemma4}, Lemma \ref{lemma6}, Assumption \ref{asscp}, and Lemma \ref{lemmam}, we can bound the $\infty$-norm of each term on the RHS as follows.
\begin{gather*}
  \|\widehat{G}_k^\prime(\widehat{\Omega}-\Omega)\widehat{M}_k\|_\infty \leq q^2 \|\widehat{G}_k\|_\infty \|\widehat{\Omega}-\Omega\|_\infty \|\widehat{M}_k\|_\infty = O_p(1)O_p(\varepsilon_N)O_p(1) = O_p(\varepsilon_N), \\
  \|(\widehat{G}_k-G)^\prime\Omega\widehat{M}_k\|_\infty \leq q^2 \|\widehat{G}_k-G\|_\infty \|\Omega\|_\infty \|\widehat{M}_k\|_\infty = O_p(\varepsilon_N)O(1)O_p(1) = O_p(\varepsilon_N), \\
  \|G^\prime\Omega (\widehat{M}_k-M)\|_\infty \leq q^2 \|G\|_\infty \|\Omega\|_\infty \|\widehat{M}_k-M\|_\infty = O(1)O(1)O_p(\varepsilon_N) = O_p(\varepsilon_N).
\end{gather*}
The conclusion then follows from the above results and the triangle inequality.
\end{proof}

\begin{lemma}
  \label{lemma7}
  $\|\widehat{G}_k^\prime\widehat{\Omega}\widehat{G}_k\rho_\star - G^\prime \Omega G \rho_\star\|_\infty = O_p(\varepsilon_N)$.
\end{lemma}
\begin{proof}
By Lemma \ref{lemma1}, $(\overline{\rho} - \rho_\star)^\prime G^\prime \Omega G (\overline{\rho} - \rho_\star)=O(\varepsilon_N^{4\xi/(2\xi+1)})=o(1)$. By the FOC of \eqref{eqrho}, $G^\prime \Omega M = G^\prime \Omega G \overline{\rho}$. Thus, 
\[
M^\prime \Omega M = (M-G\overline{\rho}+G\overline{\rho})^\prime \Omega (M-G\overline{\rho}+G\overline{\rho}) = (M-G\overline{\rho})^\prime \Omega (M-G\overline{\rho}) + \overline{\rho}^\prime G^\prime \Omega G \overline{\rho},
\]
implying $\overline{\rho}^\prime G^\prime \Omega G \overline{\rho} \leq M^\prime \Omega M$. By Assumption \ref{asscp} and \ref{assbd}, we have 
\[
M^\prime \Omega M \leq \lambda_{max}(\Omega) \|M\|_2^2 \leq C_\Omega q C_m^2 = O(1).
\]
Using the above results, it follows from Assumption \ref{assei} that 
\[
\begin{aligned}
  & C_L \|\rho_\star\|_2^2 \leq \lambda_{min}(G^\prime \Omega G) \|\rho_\star\|_2^2 \leq \rho_\star^\prime G^\prime \Omega G \rho_\star = (\rho_\star-\overline{\rho}+\overline{\rho})^\prime G^\prime \Omega G (\rho_\star-\overline{\rho}+\overline{\rho}) \\
  \leq & 2 (\rho_\star-\overline{\rho})^\prime G^\prime \Omega G (\rho_\star-\overline{\rho}) + 2 \overline{\rho}^\prime G^\prime \Omega G \overline{\rho} = o(1) + O(1) = O(1),
\end{aligned}
\]
which implies $\|\rho_\star\|_2 = O(1)$, leading to 
\begin{equation}
  \label{eqrs}
  \|\rho_\star\|_1 \leq \sqrt{p}\|\rho_\star\|_2 = O(1).
\end{equation}

Note that
\[ 
  \widehat{G}_k^\prime\widehat{\Omega}\widehat{G}_k - G^\prime \Omega G = \widehat{G}_k^\prime(\widehat{\Omega}-\Omega)\widehat{G}_k + (\widehat{G}_k-G)^\prime \Omega \widehat{G}_k + G^\prime \Omega (\widehat{G}_k-G).
\]
and by Assumption \ref{assbd}, $\|G\|_\infty = O(1)$ and $\|\widehat{G}\|_\infty = O_p(1)$. Therefore, by Lemma \ref{lemma4}, Lemma \ref{lemma6}, and Assumption \ref{asscp}, we can bound the $\infty$-norm of each term on the RHS as follows.
\begin{gather*}
  \|\widehat{G}_k^\prime(\widehat{\Omega}-\Omega)\widehat{G}_k\|_\infty \leq q^2 \|\widehat{G}_k\|_\infty \|\widehat{\Omega}-\Omega\|_\infty \|\widehat{G}_k\|_\infty = O_p(1)O_p(\varepsilon_N)O_p(1) = O_p(\varepsilon_N), \\
  \|(\widehat{G}_k-G)^\prime\Omega\widehat{G}_k\|_\infty \leq q^2 \|\widehat{G}_k-G\|_\infty \|\Omega\|_\infty \|\widehat{G}_k\|_\infty = O_p(\varepsilon_N)O(1)O_p(1) = O_p(\varepsilon_N), \\
  \|G^\prime \Omega (\widehat{G}_k-G)\|_\infty \leq q^2 \|G\|_\infty \|\Omega\|_\infty \|\widehat{G}_k-G\|_\infty = O(1)O(1)O_p(\varepsilon_N) = O_p(\varepsilon_N).
\end{gather*}
By the triangle inequality, these would imply that $\|\widehat{G}_k^\prime\widehat{\Omega}\widehat{G}_k - G^\prime \Omega G\|_\infty = O_p(\varepsilon_N)$. 

The result is obtained by 
\[
\begin{aligned}
  & \|\widehat{G}_k^\prime\widehat{\Omega}\widehat{G}_k\rho_\star - G^\prime \Omega G \rho_\star\|_\infty = \|(\widehat{G}_k^\prime\widehat{\Omega}\widehat{G}_k - G^\prime \Omega G) \rho_\star\|_\infty \\
  \leq & \|(\widehat{G}_k^\prime\widehat{\Omega}\widehat{G}_k - G^\prime \Omega G)\|_\infty \|\rho_\star\|_1 = O_p(\varepsilon_N) O(1) = O_p(\varepsilon_N).
\end{aligned}
\]
\end{proof}

\begin{lemma}
  \label{lemma8}
  For any $\widehat{J}$ such that $(\rho_\star)_{\widehat{J}^c}=\mathbf{0}$, with probability approaching one, we have 
  \[
  (\widehat{\rho}_k-\rho_\star)^\prime \widehat{G}_k^\prime \widehat{\Omega}\widehat{G}_k (\widehat{\rho}_k-\rho_\star) \leq 3r\|\widehat{\rho}_k-\rho_\star\|_1 \text{ and } \|(\widehat{\rho}_k-\rho_\star)_{\widehat{J}^c}\|_1 \leq 3\|(\widehat{\rho}_k-\rho_\star)_{\widehat{J}}\|_1.
  \]
\end{lemma}
\begin{proof}
  As $\widehat{\Omega}$ is positive definite, we can write 
\[
\widehat{\rho}_k = \mathop{\arg\min}_{\rho\in\mathbb{R}^q} \|\widehat{\Omega}^{1/2}(\widehat{M}_k-\widehat{G}_k\rho)\|_2^2 + 2r\|\rho\|_1.
\]
Therefore,
\begin{equation}
  \label{pf2}
  \|\widehat{\Omega}^{1/2}(\widehat{M}_k-\widehat{G}_k\widehat{\rho}_k)\|_2^2 + 2r\|\widehat{\rho}_k\|_1 \leq \|\widehat{\Omega}^{1/2}(\widehat{M}_k-\widehat{G}_k\rho_\star)\|_2^2 + 2r\|\rho_\star\|_1.
\end{equation}
Note that 
\begin{align*}
  & \|\widehat{\Omega}^{1/2}(\widehat{M}_k-\widehat{G}_k\widehat{\rho}_k)\|_2^2 - \|\widehat{\Omega}^{1/2}(\widehat{M}_k-\widehat{G}_k\rho_\star)\|_2^2 \\
  = & \widehat{\rho}_k^\prime \widehat{G}_k^\prime \widehat{\Omega} \widehat{G}_k\widehat{\rho}_k - \rho_\star^\prime \widehat{G}_k^\prime \widehat{\Omega} \widehat{G}_k\rho_\star - 2(\widehat{G}_k^\prime\widehat{\Omega}\widehat{M}_k)^\prime(\widehat{\rho}_k-\rho_\star) \\ 
  = & (\widehat{\rho}_k-\rho_\star)^\prime\widehat{G}_k^\prime\widehat{\Omega}\widehat{G}_k(\widehat{\rho}_k-\rho_\star) + 2\rho_\star^\prime\widehat{G}_k^\prime \widehat{\Omega} \widehat{G}_k(\widehat{\rho}_k-\rho_\star) - 2(\widehat{G}_k^\prime\widehat{\Omega}\widehat{M}_k)^\prime(\widehat{\rho}_k-\rho_\star) \\ 
  = & \|\widehat{\Omega}^{1/2}\widehat{G}_k (\widehat{\rho}_k-\rho_\star)\|_2^2 - 2(\widehat{G}_k^\prime\widehat{\Omega}\widehat{M}_k-\widehat{G}_k^\prime\widehat{\Omega}\widehat{G}_k\rho_\star)^\prime(\widehat{\rho}_k-\rho_\star).
\end{align*}
Thus, it follows from \eqref{pf2} that 
\begin{equation}
  \label{eqb3}
  \|\widehat{\Omega}^{1/2}\widehat{G}_k (\widehat{\rho}_k-\rho_\star)\|_2^2 + 2r\|\widehat{\rho}_k\|_1 \leq 2(\widehat{G}_k^\prime\widehat{\Omega}\widehat{M}_k-\widehat{G}_k^\prime\widehat{\Omega}\widehat{G}_k\rho_\star)^\prime(\widehat{\rho}_k-\rho_\star) + 2r\|\rho_\star\|_1.
\end{equation}

By the FOC of \eqref{eqrho}, we have $G^\prime\Omega M = G^\prime\Omega G \overline{\rho}$. Then, it follows from the triangle inequality, Lemma \ref{lemma2}, Lemma \ref{lemmagh}, and Lemma \ref{lemma7} that 
\begin{align*}
  & \|\widehat{G}_k^\prime\widehat{\Omega}\widehat{M}_k - \widehat{G}_k^\prime\widehat{\Omega}\widehat{G}_k\rho_\star\|_\infty \\
  \leq & \|\widehat{G}_k^\prime\widehat{\Omega}\widehat{M}_k - G^\prime\Omega M\|_\infty + \|G^\prime\Omega M - G^\prime \Omega G \rho_\star\|_\infty + \|G^\prime \Omega G \rho_\star - \widehat{G}_k^\prime\widehat{\Omega}\widehat{G}_k\rho_\star\|_\infty \\
  = & \|\widehat{G}_k^\prime\widehat{\Omega}\widehat{M}_k - G^\prime\Omega M\|_\infty + \|G^\prime\Omega G \overline{\rho} - G^\prime \Omega G \rho_\star\|_\infty + \|G^\prime \Omega G \rho_\star - \widehat{G}_k^\prime\widehat{\Omega}\widehat{G}_k\rho_\star\|_\infty \\
  = & O_p(\varepsilon_N).
\end{align*}
Therefore, by the H\"{o}lder's inequality we have 
\[
|(\widehat{G}_k^\prime\widehat{\Omega}\widehat{M}_k-\widehat{G}_k^\prime\widehat{\Omega}\widehat{G}_k\rho_\star)^\prime(\widehat{\rho}_k-\rho_\star)| \leq \|\widehat{G}_k^\prime\widehat{\Omega}\widehat{M}_k-\widehat{G}_k^\prime\widehat{\Omega}\widehat{G}_k\rho_\star\|_\infty \|\widehat{\rho}_k-\rho_\star\|_1 = O_p(\varepsilon_N)\|\widehat{\rho}_k-\rho_\star\|_1.
\]
By \eqref{eqb3} and $\varepsilon_N = o(r)$, with probability approaching one, 
\begin{equation}
  \label{eqb4}
  \|\widehat{\Omega}^{1/2}\widehat{G}_k (\widehat{\rho}_k-\rho_\star)\|_2^2 + 2r\|\widehat{\rho}_k\|_1 \leq O_p(\varepsilon_N)\|\widehat{\rho}_k-\rho_\star\|_1 + 2r\|\rho_\star\|_1 \leq r\|\widehat{\rho}_k-\rho_\star\|_1 + 2r\|\rho_\star\|_1.
\end{equation}
The first conclusion is then obtained by $\|\rho_\star\|_1 \leq \|\widehat{\rho}_k\|_1 + \|\widehat{\rho}_k-\rho_\star\|_1$ and subtracting $2r\|\widehat{\rho}_k\|_1$ from both sides.

Next, since $\|\widehat{\Omega}^{1/2}\widehat{G}_k (\widehat{\rho}_k-\rho_\star)\|_2^2 \geq 0$, it follows from \eqref{eqb4} that $2r\|\widehat{\rho}_k\|_1 \leq r\|\widehat{\rho}_k-\rho_\star\|_1 + 2r\|\rho_\star\|_1$, so dividing both sides by $r$ would give 
\begin{equation}
  \label{eqb5}
  2\|\widehat{\rho}_k\|_1 \leq \|\widehat{\rho}_k-\rho_\star\|_1 + 2\|\rho_\star\|_1.
\end{equation}
It follows by $(\rho_\star)_{\widehat{J}_c} = \mathbf{0}$ that 
\[
\|\widehat{\rho}_k\|_1 = \|\rho_\star+\widehat{\rho}_k-\rho_\star\|_1 = \|(\rho_\star)_{\widehat{J}} + (\widehat{\rho}_k-\rho_\star)_{\widehat{J}}\|_1 + \|(\widehat{\rho}_k-\rho_\star)_{\widehat{J}^c}\|_1
\]
and $\|\rho_\star\|_1 = \|(\rho_\star)_{\widehat{J}}\|_1$. Substituting these into \eqref{eqb5} gives 
\begin{align*}
  & 2\|(\rho_\star)_{\widehat{J}} + (\widehat{\rho}_k-\rho_\star)_{\widehat{J}}\|_1 + 2\|(\widehat{\rho}_k-\rho_\star)_{\widehat{J}^c}\|_1 \leq 2\|(\rho_\star)_{\widehat{J}}\|_1 + \|\widehat{\rho}_k-\rho_\star\|_1 \\
  = & 2\|(\rho_\star)_{\widehat{J}}\|_1 + \|(\widehat{\rho}_k-\rho_\star)_{\widehat{J}}\|_1 + \|(\widehat{\rho}_k-\rho_\star)_{\widehat{J}^c}\|_1 \\
  \leq & 2(\|(\rho_\star)_{\widehat{J}}+(\widehat{\rho}_k-\rho_\star)_{\widehat{J}}\|_1 + \|(\widehat{\rho}_k-\rho_\star)_{\widehat{J}}\|_1) + \|(\widehat{\rho}_k-\rho_\star)_{\widehat{J}}\|_1 + \|(\widehat{\rho}_k-\rho_\star)_{\widehat{J}^c}\|_1 \\
  = & 2\|(\rho_\star)_{\widehat{J}}+(\widehat{\rho}_k-\rho_\star)_{\widehat{J}}\|_1 + 3\|(\widehat{\rho}_k-\rho_\star)_{\widehat{J}}\|_1 + \|(\widehat{\rho}_k-\rho_\star)_{\widehat{J}^c}\|_1,
\end{align*}
which implies $\|(\widehat{\rho}_k-\rho_\star)_{\widehat{J}^c}\|_1 \leq 3\|(\widehat{\rho}_k-\rho_\star)_{\widehat{J}}\|_1$ after rearranging the terms.
\end{proof}

\begin{lemma}
  \label{lemma9}
  Define $J$ to be the vector of indices of nonzero elements of $\rho_\star$, we have $|J|\leq C^*\varepsilon_N^{-2/(2\xi+1)}$, where $C^* = (64 C_U C_\Omega q C_d^2+1) C_1$.
\end{lemma}
\begin{proof}
  For all $j\in J\backslash J_0$, the FOC of \eqref{eqdrs} implies that $|e_j G^\prime \Omega G(\rho_\star-\overline{\rho})| = \varepsilon_N$. Therefore, it follows that 
  \[
  \sum_{j\in J\backslash J_0} \big[e_j G^\prime \Omega G(\rho_\star-\overline{\rho})\big]^2 = \varepsilon_N^2 |J\backslash J_0|.
  \]
  In addition, 
  \begin{align*}
    \sum_{j\in J\backslash J_0} \big[e_j G^\prime \Omega G(\rho_\star-\overline{\rho})\big]^2 & \leq \sum_{j=1}^p \big[e_j G^\prime \Omega G(\rho_\star-\overline{\rho})\big]^2 \\
    & = (\rho_\star-\overline{\rho})^\prime G^\prime \Omega G \Big(\sum_{j=1}^p e_je_j^\prime \Big) G^\prime \Omega G (\rho_\star-\overline{\rho}) \\
    & = (\rho_\star-\overline{\rho})^\prime (G^\prime \Omega G)^2 (\rho_\star-\overline{\rho}) \\
    & \leq \lambda_{max}(G^\prime \Omega G) (\rho_\star-\overline{\rho})^\prime G^\prime \Omega G (\rho_\star-\overline{\rho}) \\
    & \leq 64 C_U C_\Omega q C_d^2 C_1 \varepsilon_N^{4\xi/(2\xi+1)},
  \end{align*}
  where the last inequality follows by Lemma \ref{lemma1} and Assumption \ref{assei}. Combining the above two inequalities, we obtain $\varepsilon_N^2 |J\backslash J_0| \leq 64 C_U C_\Omega q C_d^2 C_1\varepsilon_N^{4\xi/(2\xi+1)}$. Dividing both sides by $\varepsilon_N^2$ gives $|J\backslash J_0| \leq 64 C_U C_\Omega q C_d^2 C_1\varepsilon_N^{-2/(2\xi+1)}$. By $s_0 \leq C_1 \varepsilon_N^{-2/(2\xi+1)}$, we obtain 
  \[
  \begin{aligned}
    & |J| \leq |J_0| + |J\backslash J_0| \leq s_0 + 64 C_U C_\Omega q C_d^2 C_1 \varepsilon_N^{-2/(2\xi+1)} \\
    \leq & C_1 \varepsilon_N^{-2/(2\xi+1)} + 64 C_U C_\Omega q C_d^2 C_1 \varepsilon_N^{-2/(2\xi+1)} = C^*\varepsilon_N^{-2/(2\xi+1)}.
  \end{aligned}
  \]
\end{proof}

\begin{lemma}{\rm(Lemma 6.9 of \citet{buhlmann2011statistics})}
  \label{lemmabuh}
  Let $b_1 \ge b_2 \ge \dots \ge 0$. For $1 < q < \infty$ and $s \in \{1, \dots\}$, we have
\[
\left(\sum_{j \ge s+1} b_j^q \right)^{1/q} \le \sum_{k=1}^{\infty} \left( \sum_{j = ks+1}^{(k+1)s} b_j^q \right)^{1/q} \le s^{-(q-1)/q} \| b \|_1.
\]
\end{lemma}
\begin{proof}
  See the proof of Lemma 6.9 in \citet{buhlmann2011statistics}.
\end{proof}

\begin{lemma}
  \label{lemma10}
  $\|\widehat{\rho}_k-\rho_\star\|_2^2 = O_p(r^2\varepsilon_N^{-2/(2\xi+1)})$.
\end{lemma}
\begin{proof}
  With probability approaching one, we have the following conclusions. By the definition of $J$, we have $(\rho_\star)_{J^c} = \mathbf{0}$, so in Lemma \ref{lemma8} we set $\widehat{J} = J$ to obtain $\|(\widehat{\rho}_k-\rho_\star)_{J^c}\|_1 \leq 3\|(\widehat{\rho}_k-\rho_\star)_{J}\|_1$. By Lemma \ref{lemma9}, we also have $|J| \leq C^*\varepsilon_N^{-2/(2\xi+1)}$. Thus, by the sparse eigenvalue condition in Assumption \ref{assei} and the first conclusion of Lemma \ref{lemma8}, 
  \[
  \begin{aligned}
    & \|(\widehat{\rho}_k-\rho_\star)_{J}\|_2^2 \leq C_E^{-1}(\widehat{\rho}_k-\rho_\star)^\prime \widehat{G}_k^\prime \widehat{\Omega}\widehat{G}_k(\widehat{\rho}_k-\rho_\star) \leq 3rC_E^{-1}\|\widehat{\rho}_k-\rho_\star\|_1 \\
    \leq & 3rC_E^{-1}(\|(\widehat{\rho}_k-\rho_\star)_{J}\|_1+\|(\widehat{\rho}_k-\rho_\star)_{J^c}\|_1) \leq 3rC_E^{-1}(\|(\widehat{\rho}_k-\rho_\star)_{J}\|_1+ 3\|(\widehat{\rho}_k-\rho_\star)_{J}\|_1)\\
    = & 12rC_E^{-1}\|(\widehat{\rho}_k-\rho_\star)_{J}\|_1 \leq 12rC_E^{-1} \sqrt{|J|} \|(\widehat{\rho}_k-\rho_\star)_{J}\|_2 \leq 12rC_E^{-1}\sqrt{C^*}\varepsilon_N^{-1/(2\xi+1)}\|(\widehat{\rho}_k-\rho_\star)_{J}\|_2.
  \end{aligned}
  \]
  Dividing both sides of the above inequality by $\|(\widehat{\rho}_k-\rho_\star)_{J}\|_2$ gives
  \[
  \|(\widehat{\rho}_k-\rho_\star)_{J}\|_2 \leq 12rC_E^{-1}\sqrt{C^*}\varepsilon_N^{-1/(2\xi+1)}.
  \]

  Next, let $N$ denote the indices corresponding to the largest $|J|$ entries in $(\widehat{\rho}_k-\rho_\star)_{J^c}$, so that $N \subset J^c$ and $|N| = |J|$. Define $\widetilde{J} = J \cup N$, we have $|\widetilde{J}| = |J| + |N| = 2|J|$ and $\widetilde{J}^c \subset J^c$. Thus, we have $(\rho_\star)_{\widetilde{J}^c}=\mathbf{0}$ and we can set $\widehat{J} = \widetilde{J}$. Then, it follows by Lemma \ref{lemma8} that $\|(\widehat{\rho}_k-\rho_\star)_{\widetilde{J}^c}\|_1 \leq 3\|(\widehat{\rho}_k-\rho_\star)_{\widetilde{J}}\|_1$. And by Lemma \ref{lemma9}, we have $|\widetilde{J}|\leq 2C^*\varepsilon_N^{-2/(2\xi+1)}$. Thus, by Assumption \ref{assei} and the first conclusion of Lemma \ref{lemma8}, 
  \[
  \begin{aligned}
    & \|(\widehat{\rho}_k-\rho_\star)_{\widetilde{J}}\|_2^2 \leq C_E^{-1}(\widehat{\rho}_k-\rho_\star)^\prime \widehat{G}_k^\prime \widehat{\Omega}\widehat{G}_k(\widehat{\rho}_k-\rho_\star) \leq 3rC_E^{-1}\|\widehat{\rho}_k-\rho_\star\|_1 \\
    \leq & 3rC_E^{-1}(\|(\widehat{\rho}_k-\rho_\star)_{\widetilde{J}}\|_1+\|(\widehat{\rho}_k-\rho_\star)_{\widetilde{J}^c}\|_1) \leq 3rC_E^{-1}(\|(\widehat{\rho}_k-\rho_\star)_{\widetilde{J}}\|_1+ 3\|(\widehat{\rho}_k-\rho_\star)_{\widetilde{J}}\|_1)\\
    = & 12rC_E^{-1}\|(\widehat{\rho}_k-\rho_\star)_{\widetilde{J}}\|_1 \leq 12rC_E^{-1} \sqrt{|\widetilde{J}|} \|(\widehat{\rho}_k-\rho_\star)_{\widetilde{J}}\|_2 \\
    \leq & 12rC_E^{-1}\sqrt{2C^*}\varepsilon_N^{-1/(2\xi+1)}\|(\widehat{\rho}_k-\rho_\star)_{\widetilde{J}}\|_2.
  \end{aligned}
  \]
  Dividing both sides of the above inequality by $\|(\widehat{\rho}_k-\rho_\star)_{\widetilde{J}}\|_2$ gives 
  \[
  \|(\widehat{\rho}_k-\rho_\star)_{\widetilde{J}}\|_2 \leq 12rC_E^{-1}\sqrt{2C^*}\varepsilon_N^{-1/(2\xi+1)},
  \]
  that is, $\|(\widehat{\rho}_k-\rho_\star)_{\widetilde{J}}\|_2 = O_p(r\varepsilon_N^{-1/(2\xi+1)})$. 

  Then, in Lemma \ref{lemmabuh}, we set $b = (\widehat{\rho}_k-\rho_\star)_{J^c}$ and $s = |N| = |J|$. Note that $J^c\backslash N = \widetilde{J}^c$, it follows by Lemma \ref{lemmabuh} that 
  \[
  \begin{aligned}
    & \|(\widehat{\rho}_k-\rho_\star)_{\widetilde{J}^c}\|_2 \leq |J|^{-1/2} \|(\widehat{\rho}_k-\rho_\star)_{J^c}\|_1 \leq |J|^{-1/2} 3\|(\widehat{\rho}_k-\rho_\star)_{J}\|_1 \\
    \leq & |J|^{-1/2} 3 \sqrt{|J|} \|(\widehat{\rho}_k-\rho_\star)_{J}\|_2 \leq 36rC_E^{-1}\sqrt{C^*}\varepsilon_N^{-1/(2\xi+1)},
  \end{aligned}
  \]
  which implies $\|(\widehat{\rho}_k-\rho_\star)_{\widetilde{J}^c}\|_2 = O_p(r\varepsilon_N^{-1/(2\xi+1)})$.

  Finally, combining the above two inequalities, we have 
  \[
  \begin{aligned}
    & \|\widehat{\rho}_k-\rho_\star\|_2^2 = \|(\widehat{\rho}_k-\rho_\star)_{\widetilde{J}}\|_2^2 + \|(\widehat{\rho}_k-\rho_\star)_{\widetilde{J}^c}\|_2^2 \\
    \leq & O_p(r^2\varepsilon_N^{-2/(2\xi+1)}) + O_p(r^2\varepsilon_N^{-2/(2\xi+1)}) = O_p(r^2\varepsilon_N^{-2/(2\xi+1)}).
  \end{aligned}
  \]
\end{proof}

\begin{proof}[Proof of Theorem \ref{th1}]
By Lemma \ref{lemma10} and Assumption \ref{assei}, 
\[
\begin{aligned}
  & \int [b(z_{it})^\prime(\widehat{\rho}_k-\rho_\star)]^2 F_0(\rd z_{it}) = (\widehat{\rho}_k-\rho_\star)^\prime \E[b(Z_{it})b(Z_{it})^\prime](\widehat{\rho}_k-\rho_\star) = (\widehat{\rho}_k-\rho_\star)^\prime B (\widehat{\rho}_k-\rho_\star) \\
  \leq & \lambda_{max}(B)\|\widehat{\rho}_k-\rho_\star\|_2^2 \leq C_B \|\widehat{\rho}_k-\rho_\star\|_2^2 = O_p(r^2\varepsilon_N^{-2/(2\xi+1)}) = O_p((r/\varepsilon_N)^2\varepsilon_N^{4\xi/(2\xi+1)}).
\end{aligned}
\]

Then, by the triangle inequality, Lemma \ref{lemma3}, and $\varepsilon_N = o(r)$,
\begin{align*}
  & \int [\widehat{\alpha}_k(Z_{it}) - \alpha_0(Z_{it})]^2 F_0(\rd z_{it}) \\
  = & \int [b(Z_{it})^\prime\widehat{\rho}_k - b(Z_{it})^\prime \rho_\star + b(Z_{it})^\prime \rho_\star - \alpha_0(Z_{it})]^2 F_0(\rd z_{it}) \\
  \leq & 2\int [b(Z_{it})^\prime(\widehat{\rho}_k-\rho_\star)]^2 F_0(\rd z_{it}) + 2\int [\alpha_0(Z_{it}) - b(Z_{it})^\prime \rho_\star]^2 F_0(\rd z_{it}) \\
  = & O_p((r/\varepsilon_N)^2\varepsilon_N^{4\xi/(2\xi+1)}) + O(\varepsilon_N^{4\xi/(2\xi+1)}) = O_p((r/\varepsilon_N)^2\varepsilon_N^{4\xi/(2\xi+1)}) = O_p(r^2\varepsilon_N^{-2/(2\xi+1)}).
\end{align*}
\end{proof}

\subsection{Proof of Theorem \ref{th2} and Theorem \ref{prop1}}
We first prove some lemmas. 
\begin{lemma}
  \label{lemma13}
  If Assumption \ref{asscp} to Assumption \ref{assei} are satisfied and $\varepsilon_N = o(r)$, for all $k$,
  \[
  \|\widehat{\rho}_k\|_1 = O_p(1).
  \]
\end{lemma}
\begin{proof}
  By \eqref{eqb4} and $\varepsilon_N = o(r)$, 
  \[
  \|\widehat{\Omega}^{1/2}\widehat{G}_k (\widehat{\rho}_k-\rho_\star)\|_2^2 + 2r\|\widehat{\rho}_k\|_1 \leq O_p(\varepsilon_N)\|\widehat{\rho}_k-\rho_\star\|_1 + 2r\|\rho_\star\|_1 \leq o_p(r)\|\widehat{\rho}_k-\rho_\star\|_1 + 2r\|\rho_\star\|_1,
  \]
  which implies 
  \[
  2r\|\widehat{\rho}_k\|_1 \leq ro_p(1)\|\widehat{\rho}_k-\rho_\star\|_1 + 2r\|\rho_\star\|_1 \leq ro_p(1)(\|\widehat{\rho}_k\|_1 + \|\rho_\star\|_1) + 2r\|\rho_\star\|_1.
  \]
  Dividing both sides by $2r$, we obtain $\|\widehat{\rho}_k\|_1 \leq o_p(1)(\|\widehat{\rho}_k\|_1 + \|\rho_\star\|_1) + \|\rho_\star\|_1$, so with probability approaching one, 
  \[
  \|\widehat{\rho}_k\|_1 \leq \frac{1}{2}(\|\widehat{\rho}_k\|_1 + \|\rho_\star\|_1) + \|\rho_\star\|_1.
  \]
  Rearranging the terms and noting that $\|\rho_\star\|_1 = O(1)$ by \eqref{eqrs} complete the proof.
\end{proof}
\begin{lemma}
  If Assumption \ref{assa} to \ref{asscr} are satisfied, we have 
  \[
  \begin{aligned}
    & \frac{1}{\sqrt{N}}\sum_{k=1}^K \sum_{i\in\mathcal{F}_k} \Big\{ m(W_{it}^*,\Delta \widehat{\gamma}_k^*) + \widehat{\alpha}_k(Z_{it})[\Delta Y_{it}^* - \Delta \widehat{\gamma}_k^*(V_{it})] \Big\} \\
    = & \frac{1}{\sqrt{N}}\sum_{k=1}^K \sum_{i\in\mathcal{F}_k} \Big\{ m(W_{it}^*,\Delta \gamma_0^*) + \alpha_0(Z_{it})[\Delta Y_{it}^* - \Delta \gamma_0^*(V_{it})] \Big\} + o_p(1).
  \end{aligned}
  \]
\end{lemma}
\begin{proof}
For all $k\in\{1,\ldots,K\}$ and $i\in\mathcal{F}_k$, define
\begin{gather*}
\widehat{R}_{1ki} = m(W_{it}^*,\Delta \widehat{\gamma}_k^*) - m(W_{it}^*,\Delta \gamma_0^*)= \frac{\partial \widehat{\gamma}_k(V_{it})}{\partial D_{i,t-s}} - \frac{\partial \gamma_0(V_{it})}{\partial D_{i,t-s}},\\
\widehat{R}_{2ki} = \alpha_0(Z_{it})[\Delta \gamma_0^*(V_{it}) - \Delta \widehat{\gamma}_k^*(V_{it})],\\
\widehat{R}_{3ki} = [\widehat{\alpha}_k(Z_{it})-\alpha_0(Z_{it})][\Delta Y_{it}^* - \Delta \gamma_0^*(V_{it})], \\
\widehat{\delta}_{i} = [\widehat{\alpha}_k(Z_{it})-\alpha_0(Z_{it})][\Delta \gamma_0^*(V_{it}) - \Delta \widehat{\gamma}_k^*(V_{it})].
\end{gather*}
It can be verified that 
\begin{align}
  & \{m(W_{it}^*,\Delta \widehat{\gamma}_k^*) + \widehat{\alpha}_k(Z_{it})[\Delta Y_{it}^* - \Delta \widehat{\gamma}_k^*(V_{it})]\} - \{m(W_{it}^*,\Delta \gamma_0^*) + \alpha_0(Z_{it})[\Delta Y_{it}^* - \Delta \gamma_0^*(V_{it})]\} \notag \\
  = & \widehat{R}_{1ki} + \widehat{R}_{2ki} + \widehat{R}_{3ki} + \widehat{\delta}_i. \label{final3}
\end{align}

Define $\wc\coloneq(W_{it})_{i\in\mathcal{F}_{k,k^\prime}^c}$. Note that we use $\wc$ to construct $\widehat{\gamma}_k$, so $\widehat{\gamma}_k$ is a function of $\wc$. Therefore, by the fact that given $\gamma_0$ and $\widehat{\gamma}_k$, $\widehat{R}_{1ki}$ depends only on $V_{it}$ and the i.i.d. assumption, we find that $\forall i,j\in\mathcal{F}_k, i\neq j$, $\widehat{R}_{1ki}$ and $\widehat{R}_{1kj}$ are conditionally independent given $\wc$. Consequently, 
\begin{align*}
  & \E\bigg(\bigg[\frac{1}{\sqrt{N}}\sum_{i\in\mathcal{F}_k}\left(\widehat{R}_{1ki}-\E(\widehat{R}_{1ki}\mid \wc)\right)\bigg]^2\Big| \wc \bigg) = \frac{n_k}{N} \Var(\widehat{R}_{1ki} \mid \wc) \\
  \leq & \E(\widehat{R}_{1ki}^2 \mid \wc) = \int \bigg[\frac{\partial \widehat{\gamma}_k(v_{it})}{\partial D_{i,t-s}} - \frac{\partial \gamma_0(v_{it})}{\partial D_{i,t-s}}\bigg]^2 F_0(\rd v_{it}) \pto 0,
\end{align*}
where the convergence in probability follows from Assumption \ref{assm}. Then, by the conditional Markov inequality, $\forall \varepsilon$,
\[
P_N = \rP\bigg(\bigg|\frac{1}{\sqrt{N}}\sum_{i\in\mathcal{F}_k}\left(\widehat{R}_{1ki}-\E(\widehat{R}_{1ki}\mid \wc)\right)\bigg| > \varepsilon \Big| \wc \bigg) \pto 0.
\]
Since the sequence $\{|P_N|\}$ is uniformly integrable, $P_N \pto 0$ implies $\E(|P_N|) \to 0$. Therefore, 
\begin{equation}
\label{equi}
\rP\bigg(\bigg|\frac{1}{\sqrt{N}}\sum_{i\in\mathcal{F}_k}\left(\widehat{R}_{1ki}-\E(\widehat{R}_{1ki}\mid \wc)\right)\bigg| > \varepsilon \bigg)=\E(P_N)=\E(|P_N|) \to 0,  
\end{equation}
that is,
\begin{equation}
\label{cp1}
\frac{1}{\sqrt{N}}\sum_{i\in\mathcal{F}_k}\left(\widehat{R}_{1ki}-\E(\widehat{R}_{1ki}\mid \wc)\right) \pto 0.
\end{equation}

Next, we investigate $\widehat{R}_{2ki}$. Note that given $\gamma_0$ and $\widehat{\gamma}_k$, $\widehat{R}_{2ki}$ depends on $V_{it}$ and $(V_{mt})_{m\in\mathcal{F}_{k^\prime}}$, so $\forall i,j \in \mathcal{F}_k$, $\widehat{R}_{2ki}$ and $\widehat{R}_{2kj}$ are not conditionally independent. We have 
\begin{align*}
& \E\bigg(\bigg[\frac{1}{\sqrt{N}}\sum_{i\in\mathcal{F}_k}\left(\widehat{R}_{2ki}-\E(\widehat{R}_{2ki}\mid \wc)\right)\bigg]^2\Big| \wc \bigg)\\
& = \frac{n_k}{N} \Var(\widehat{R}_{2ki} \mid \wc) + \frac{n_k(n_k-1)}{N}\Cov(\widehat{R}_{2ki},\widehat{R}_{2kj} \mid \wc).
\end{align*}
Define $A_i = \Delta \gamma_0(V_{it})-\Delta \widehat{\gamma}_k(V_{it})$. By definition, $\widehat{R}_{2ki}=\alpha_0(Z_{it})(A_i-\frac{1}{n_{k^\prime}}\sum_{u\in\mathcal{F}_{k^\prime}}A_u)$. Therefore, we can write 
\begin{align*}
& \Cov(\widehat{R}_{2ki},\widehat{R}_{2kj} \mid \wc)\\
= & \Cov(\alpha_0(Z_{it})A_i,\alpha_0(Z_{jt})A_j \mid \wc) - \Cov\Big(\alpha_0(Z_{it})A_i,\alpha_0(Z_{jt})\frac{1}{n_{k^\prime}}\sum_{u\in\mathcal{F}_{k^\prime}}A_u \mid \wc\Big) \\ 
& - \Cov\Big(\alpha_0(Z_{jt})A_j,\alpha_0(Z_{it})\frac{1}{n_{k^\prime}}\sum_{u\in\mathcal{F}_{k^\prime}}A_u \mid \wc\Big) \\
& + \Cov\Big(\alpha_0(Z_{it})\frac{1}{n_{k^\prime}}\sum_{u\in\mathcal{F}_{k^\prime}}A_u,\alpha_0(Z_{jt})\frac{1}{n_{k^\prime}}\sum_{u\in\mathcal{F}_{k^\prime}}A_u \mid \wc\Big) \\
= & \E\Big(\alpha_0(Z_{it})\alpha_0(Z_{jt})\Big(\frac{1}{n_{k^\prime}}\sum_{u\in\mathcal{F}_{k^\prime}}A_u\Big)^2 \mid \wc \Big) - \Big[\E\Big(\alpha_0(Z_{it})\Big(\frac{1}{n_{k^\prime}}\sum_{u\in\mathcal{F}_{k^\prime}}A_u\Big)\mid \wc \Big)\Big]^2\\
= & [\E(\alpha_0(Z_{it})\mid\wc)]^2 \E\Big[\Big(\frac{1}{n_{k^\prime}}\sum_{u\in\mathcal{F}_{k^\prime}}A_u\Big)^2 \mid \wc \Big] \\
& - [\E(\alpha_0(Z_{it})\mid\wc)]^2 \Big[\E\Big(\frac{1}{n_{k^\prime}}\sum_{u\in\mathcal{F}_{k^\prime}}A_u \mid \wc \Big)\Big]^2 \\
= & [\E(\alpha_0(Z_{it}))]^2 \frac{1}{n_{k^\prime}}\Var(A_i \mid \wc),
\end{align*}
where the second, third, and fourth inequalities hold because of the i.i.d. assumption and the fact that $\widehat{\gamma}_k$ is a function of $\wc$. Thus, 
\begin{align*}
& \E\bigg(\bigg[\frac{1}{\sqrt{N}}\sum_{i\in\mathcal{F}_k}\left(\widehat{R}_{2ki}-\E(\widehat{R}_{2ki}\mid \wc)\right)\bigg]^2\Big| \wc \bigg) \\
= & \frac{n_k}{N} \Var(\widehat{R}_{2ki} \mid \wc) + \frac{n_k(n_k-1)}{Nn_{k^\prime}} [\E(\alpha_0(Z_{it}))]^2 \Var(A_i \mid \wc) \\
\leq & \E(\widehat{R}_{2ki}^2\mid \wc) + [\E(\alpha_0(Z_{it}))]^2 \E(A_i^2 \mid \wc) \\
\leq & C_\alpha^2 \E([\Delta \gamma_0^*(V_{it}) - \Delta \widehat{\gamma}_k^*(V_{it})]^2 \mid \wc) + C_\alpha^2 \E(A_i^2 \mid \wc)\\
= & C_\alpha^2 \int [\Delta \widehat{\gamma}_k^*(v_{it}) - \Delta \gamma_0^*(v_{it})]^2 F_0(\rd w_{it}^*) + C_\alpha^2\int [\Delta \widehat{\gamma}_k(v_{it}) - \Delta \gamma_0(v_{it})]^2 F_0(\rd \overline{v}_{i,t-1:t})\\
& \hspace{-.5cm}\pto 0,
\end{align*}
where the convergence in probability holds because of Assumption \ref{assm} and the fact that $\int [\Delta \widehat{\gamma}_k(v_{it}) - \Delta \gamma_0(v_{it})]^2 F_0(\rd \overline{v}_{i,t-1:t}) \pto 0$ implies $\int [\Delta \widehat{\gamma}_k^*(v_{it}) - \Delta \gamma_0^*(v_{it})]^2 F_0(\rd w_{it}^*) \pto 0$ by the definition of $\widehat{\gamma}_k^*$ and $\gamma_0^*$, the i.i.d. assumption, and the Cauchy-Schwarz inequality. Then, similar to the derivation of \eqref{cp1}, it can be proved that 
\begin{equation}
  \label{cp2}
  \frac{1}{\sqrt{N}}\sum_{i\in\mathcal{F}_k}\left(\widehat{R}_{2ki}-\E(\widehat{R}_{2ki}\mid \wc)\right) \pto 0.
\end{equation}

For $\widehat{R}_{3ki}$, we have 
\begin{align*}
& \E\bigg(\bigg[\frac{1}{\sqrt{N}}\sum_{i\in\mathcal{F}_k}\left(\widehat{R}_{3ki}-\E(\widehat{R}_{3ki}\mid \wc)\right)\bigg]^2\Big| \wc \bigg)\\
& = \frac{n_k}{N} \Var(\widehat{R}_{3ki} \mid \wc) + \frac{n_k(n_k-1)}{N}\Cov(\widehat{R}_{3ki},\widehat{R}_{3kj} \mid \wc).
\end{align*}
Note that $\widehat{R}_{3ki} = [\widehat{\alpha}_k(Z_{it})-\alpha_0(Z_{it})]\Delta \varepsilon_{it}^*$, similar to the calculation of $\Cov(\widehat{R}_{2ki},\widehat{R}_{2kj} \mid \wc)$, we have 
\begin{align*}
& \Cov(\widehat{R}_{3ki},\widehat{R}_{3kj} \mid \wc) = \big(\E[\widehat{\alpha}_k(Z_{it})-\alpha_0(Z_{it}) \mid \wc]\big)^2 \frac{1}{n_{k^\prime}}\Var(\Delta \varepsilon_{it})\\
\leq & \E\big[\big(\widehat{\alpha}_k(Z_{it})-\alpha_0(Z_{it})\big)^2 \mid \wc\big]\frac{1}{n_{k^\prime}}\E[(\Delta \varepsilon_{it})^2].
\end{align*}
Moreover, 
\begin{align*}
  & \E(\widehat{R}_{3ki}^2 \mid \wc) = \E[[\widehat{\alpha}_k(Z_{it})-\alpha_0(Z_{it})]^2[\Delta Y_{it}^* - \Delta \gamma_0^*(V_{it})]^2 \mid \wc] \\
  = & \E(\widehat{R}_{3ki}^2 \mid \wc) = \E[[\widehat{\alpha}_k(Z_{it})-\alpha_0(Z_{it})]^2\E\{[\Delta Y_{it}^* - \Delta \gamma_0^*(V_{it})]^2 \mid Z_{it}, \wc\} \mid \wc],
\end{align*}
where by Assumption \ref{assa} $\E\{[\Delta Y_{it}^* - \Delta \gamma_0^*(V_{it})]^2 \mid Z_{it}, \wc\} = \E\{[\Delta Y_{it}^* - \Delta \gamma_0^*(V_{it})]^2 \mid Z_{it}\} \leq 4C_z$. This is because  
\begin{align*}
  & \E\{[\Delta Y_{it}^* - \Delta \gamma_0^*(V_{it})]^2 \mid Z_{it}\} = \E\Big[\Big(\Delta \varepsilon_{it}-\frac{1}{n_{k^\prime}}\sum_{m\in\mathcal{F}_{k^\prime}}\Delta \varepsilon_{mt}\Big)^2 \mid Z_{it}\Big] \\
  \leq & 2\E[(\Delta \varepsilon_{it})^2 \mid Z_{it}] + 2\E\Big[\Big(\frac{1}{n_{k^\prime}}\sum_{m\in\mathcal{F}_{k^\prime}}\Delta \varepsilon_{mt}\Big)^2 \mid Z_{it}\Big] \\
  = & 2\E[(\Delta \varepsilon_{it})^2 \mid Z_{it}] + 2\E\Big[\Big(\frac{1}{n_{k^\prime}}\sum_{m\in\mathcal{F}_{k^\prime}}\Delta \varepsilon_{mt}\Big)^2\Big] \leq 2\E[(\Delta \varepsilon_{it})^2 \mid Z_{it}] + 2\frac{1}{n_{k^\prime}}\sum_{m\in\mathcal{F}_{k^\prime}}\E[(\Delta \varepsilon_{mt})^2]
\end{align*} 
and Assumption \ref{assa} implies $\E[(\Delta \varepsilon_{it})^2] \leq C_z$.

Hence, 
\begin{align*}
& \E\bigg(\bigg[\frac{1}{\sqrt{N}}\sum_{i\in\mathcal{F}_k}\left(\widehat{R}_{3ki}-\E(\widehat{R}_{3ki}\mid \wc)\right)\bigg]^2\Big| \wc \bigg)\\
\leq & \E(\widehat{R}_{3ki}^2\mid \wc) + \E[[\widehat{\alpha}_k(Z_{it})-\alpha_0(Z_{it})]^2 \mid \wc]\E[(\Delta \varepsilon_{it})^2] \\
\leq & 4C_z \E[[\widehat{\alpha}_k(Z_{it})-\alpha_0(Z_{it})]^2 \mid \wc] + \E[[\widehat{\alpha}_k(Z_{it})-\alpha_0(Z_{it})]^2 \mid \wc] C_z\\
= & 5C_z \int [\widehat{\alpha}_k(z_{it}) - \alpha_0(z_{it})]^2 F_0(\rd z_{it}) \pto 0,
\end{align*}
where the convergence in probability follows by Assumption \ref{assm}. Similar to the derivation of \eqref{cp1}, this would imply that 
\begin{equation}
  \label{cp3}
  \frac{1}{\sqrt{N}}\sum_{i\in\mathcal{F}_k}\left(\widehat{R}_{3ki}-\E(\widehat{R}_{3ki}\mid \wc)\right) \pto 0.
\end{equation}

Note that by the i.i.d. assumption, for any non-random function $\Delta \gamma^*$, \eqref{eqrie} would imply \\$\E[m(W_{it}^*,\Delta \gamma^*) - \alpha_0(Z_{it})\Delta \gamma^*(V_{it}) \mid \wc] = 0$. Thus, because $\Delta \widehat{\gamma}_k^*$ is a non-random function conditioned on $\wc$, we have 
\begin{equation}
  \label{cp4}
  \begin{aligned}
  & \E[\widehat{R}_{1ki} + \widehat{R}_{2ki} \mid \wc] = \E[m(W_{it}^*,\Delta \widehat{\gamma}_k^*)-\alpha_0(Z_{it})\Delta \widehat{\gamma}_k^*(V_{it})\mid \wc] \\
  & - \E[m(W_{it}^*,\Delta \gamma_0^*)-\alpha_0(Z_{it})\Delta \gamma_0^*(V_{it}) \mid \wc] = 0.
  \end{aligned}
\end{equation}
In addition, by \eqref{eqer},
\begin{equation}
  \label{cp5}
  \begin{aligned}
  & \E[\widehat{R}_{3ki} \mid \wc] = \E[[\widehat{\alpha}_k(Z_{it})-\alpha_0(Z_{it})][\Delta Y_{it}^* - \Delta \gamma_0^*(V_{it})] \mid \wc] \\
  = & \E[[\widehat{\alpha}_k(Z_{it})-\alpha_0(Z_{it})] \E[\Delta \varepsilon_{it}^* \mid \wc,Z_{it}] \mid \wc] \\
  = & \E[[\widehat{\alpha}_k(Z_{it})-\alpha_0(Z_{it})] \E[\Delta \varepsilon_{it}^* \mid Z_{it}] \mid \wc] = 0.
  \end{aligned}
\end{equation} 
Consequently, it follows from \eqref{cp1}, \eqref{cp2}, \eqref{cp3}, \eqref{cp4}, and \eqref{cp5} that 
\begin{equation}
\label{final1}
\frac{1}{\sqrt{N}}\sum_{i\in\mathcal{F}_k}\left(\widehat{R}_{1ki}+\widehat{R}_{2ki}+\widehat{R}_{3ki}\right) \pto 0.
\end{equation}

Next, consider $\widehat{\delta}_i$. We have 
\begin{align*}
& \Cov(\widehat{\delta}_i, \widehat{\delta}_j \mid \wc) = \big(\E[\widehat{\alpha}_k(Z_{it}) - \alpha_0(Z_{it}) \mid \wc]\big)^2 \frac{1}{n_{k^\prime}}\Var(A_i\mid\wc)\\
& \leq \E\big[\big(\widehat{\alpha}_k(Z_{it}) - \alpha_0(Z_{it})\big)^2 \mid \wc\big]\frac{1}{n_{k^\prime}}\E(A_i^2\mid\wc),
\end{align*}
which implies 
\begin{align}
& \E\bigg(\bigg[\frac{1}{\sqrt{N}}\sum_{i\in\mathcal{F}_k}\left(\widehat{\delta}_i-\E(\widehat{\delta}_i \mid \wc)\right)\bigg]^2\Big| \wc \bigg) \notag\\
= & \frac{n_k}{N} \Var(\widehat{\delta}_i \mid \wc) + \frac{n_k(n_k-1)}{N}\Cov(\widehat{\delta}_i,\widehat{\delta}_j \mid \wc) \notag\\
\leq & \E(\widehat{\delta}_i^2\mid\wc) + \E\big[\big(\widehat{\alpha}_k(Z_{it}) - \alpha_0(Z_{it})\big)^2 \mid \wc\big]\E(A_i^2\mid\wc) \notag\\
= & \int [\widehat{\alpha}_k(z_{it}) - \alpha_0(z_{it})]^2[\Delta \widehat{\gamma}_k^*(v_{it}) - \Delta \gamma_0^*(v_{it})]^2 F_0(\rd w_{it}^*) \label{l1}\\
& + \int [\widehat{\alpha}_k(z_{it}) - \alpha_0(z_{it})]^2 F_0(\rd z_{it}) \int [\Delta \widehat{\gamma}_k(v_{it}) - \Delta \gamma_0(v_{it})]^2 F_0(\rd \overline{v}_{i,t-1:t}). \label{l2}
\end{align}
By Assumption \ref{assbd} and Lemma \ref{lemma13}, $\widehat{\alpha}_k(Z_{it}) = O_p(1)$, which together with $\alpha_0(Z_{it}) \leq C_\alpha$ implies $[\widehat{\alpha}_k(z_{it}) - \alpha_0(z_{it})]^2 = O_p(1)$. Also, note that Assumption \ref{assm} implies $\int [\Delta \widehat{\gamma}_k^*(v_{it}) - \Delta \gamma_0^*(v_{it})]^2 F_0(\rd w_{it}^*) \pto 0$. Thus, $\eqref{l1} \pto 0$. By Assumption \ref{assm}, $\eqref{l2} \pto 0$. Therefore, similar to the derivation of \eqref{cp1}, 
\begin{equation}
  \label{cp6}
  \frac{1}{\sqrt{N}}\sum_{i\in\mathcal{F}_k} \Big(\widehat{\delta}_i - \E\Big(\widehat{\delta}_i \mid \wc\Big)\Big) \pto 0.
\end{equation}

Note that 
\begin{align*}
  & \frac{1}{\sqrt{N}} \sum_{i\in\mathcal{F}_k} \E\Big(\widehat{\delta}_i \mid \wc\Big)= \frac{n_k}{\sqrt{N}} \E\Big(\widehat{\delta}_i \mid \wc\Big) \leq \sqrt{N} \E\Big(\widehat{\delta}_i \mid \wc\Big) \\
  = & \sqrt{N} \int [\widehat{\alpha}_k(z_{it}) - \alpha_0(z_{it})][\Delta \widehat{\gamma}_k^*(v_{it}) - \Delta \gamma_0^*(v_{it})] F_0(\rd w_{it}^*) \\
  = & \sqrt{N} \int [\widehat{\alpha}_k(z_{it}) - \alpha_0(z_{it})]\Big\{\int [\Delta \widehat{\gamma}_k^*(v_{it}) - \Delta \gamma_0^*(v_{it})]F_{0,W_{it}^* \mid z_{it}}(\rd w_{it}^*) \Big\} F_0(\rd z_{it}) \\
  \leq & \sqrt{N} \bigg\{\int [\widehat{\alpha}_k(z_{it}) - \alpha_0(z_{it})]^2 F_0(\rd z_{it})\bigg\}^{1/2} \\
  & \hspace{3cm} \bigg\{\int \Big\{\int [\Delta \widehat{\gamma}_k^*(v_{it}) - \Delta \gamma_0^*(v_{it})]F_{0,W_{it}^* \mid z_{it}}(\rd w_{it}^*) \Big\}^2 F_0(\rd z_{it})\bigg\}^{1/2},
\end{align*}
and by the Cauchy-Schwarz inequality, Assumption \ref{asscr} would imply that the last line converges to 0 in probability, where $F_{0,W_{it}^* \mid z_{it}}$ is the conditional distribution function of $W_{it}^*$ given $Z_{it}=z_{it}$. Therefore, it follows by \eqref{cp6} that 
\begin{equation}
  \label{final2}
  \frac{1}{\sqrt{N}}\sum_{i\in\mathcal{F}_k} \widehat{\delta}_i \pto 0.
\end{equation}

Finally, it follows from \eqref{final1}, \eqref{final2}, and \eqref{final3} that 
\begin{align*}
  & \frac{1}{\sqrt{N}}\sum_{k=1}^K \sum_{i\in\mathcal{F}_k} \Big\{ m(W_{it}^*,\Delta \widehat{\gamma}_k^*) + \widehat{\alpha}_k(Z_{it})[\Delta Y_{it}^* - \Delta \widehat{\gamma}_k^*(V_{it})] \Big\} \\
  & - \frac{1}{\sqrt{N}}\sum_{k=1}^K \sum_{i\in\mathcal{F}_k} \Big\{ m(W_{it}^*,\Delta \gamma_0^*) + \alpha_0(Z_{it})[\Delta Y_{it}^* - \Delta \gamma_0^*(V_{it})] \Big\} \\
  = & \sum_{k=1}^K \frac{1}{\sqrt{N}} \sum_{i\in\mathcal{F}_k} \Big(\widehat{R}_{1ki} + \widehat{R}_{2ki} + \widehat{R}_{3ki} + \widehat{\delta}_i\Big) = \sum_{k=1}^K o_p(1) = o_p(1).
\end{align*}
\end{proof}

\begin{proof}[Proof of Theorem \ref{th2}]
Note that 
\begin{align*}
  & \frac{1}{\sqrt{N}} \sum_{k=1}^K \sum_{i\in\mathcal{F}_k} \Big\{ m(W_{it}^*,\Delta \gamma_0^*) - \theta_{0t}(s) + \alpha_0(Z_{it})[\Delta Y_{it}^* - \Delta \gamma_0^*(V_{it})] \Big\} \\
  = & \frac{1}{\sqrt{N}} \sum_{k=1}^K \sum_{i\in\mathcal{F}_k} \bigg\{ \frac{\partial \gamma_0(V_{it})}{\partial D_{i,t-s}} - \theta_{0t}(s) + \alpha_0(Z_{it})\Big(\Delta \varepsilon_{it} - \frac{1}{n_{k^\prime}}\sum_{u\in\mathcal{F}_{k^\prime}}\Delta \varepsilon_{ut}\Big) \bigg\}\\
  = & \frac{1}{\sqrt{N}} \sum_{i=1}^N \bigg\{\frac{\partial \gamma_0(V_{it})}{\partial D_{i,t-s}} - \theta_{0t}(s) + \alpha_0(Z_{it})\Delta \varepsilon_{it}\bigg\} - \frac{1}{\sqrt{N}} \sum_{k=1}^K \Big\{\sum_{i\in\mathcal{F}_k} \alpha_0(Z_{it}) \frac{1}{n_{k^\prime}}\sum_{u\in\mathcal{F}_{k^\prime}}\Delta \varepsilon_{ut}\Big\}.
\end{align*}
Define $S_\alpha^{(k)} = \sum_{i\in\mathcal{F}_k}\alpha_0(Z_{it})$, $S_{\alpha c}^{(k)} = \sum_{i\in\mathcal{F}_k}[\alpha_0(Z_{it})-\E[\alpha_0(Z_{it})]]$, and $S_\varepsilon^{(k)} = \sum_{i\in\mathcal{F}_k}\Delta \varepsilon_{it}$. Thus, we can rewrite the second term as 
\begin{align*}
  & \frac{1}{\sqrt{N}} \sum_{k=1}^K \Big\{\sum_{i\in\mathcal{F}_k} \alpha_0(Z_{it}) \frac{1}{n_{k^\prime}}\sum_{u\in\mathcal{F}_{k^\prime}}\Delta \varepsilon_{ut}\Big\} = \frac{1}{\sqrt{N}} \sum_{k=1}^K \frac{1}{n_{k^\prime}}S_\alpha^{(k)} S_\varepsilon^{(k^\prime)} \\
  = & \frac{1}{\sqrt{N}} \sum_{k=1}^K \frac{1}{n_{k^\prime}}[S_{\alpha c}^{(k)}+n_k\E[\alpha_0(Z_{it})]] S_\varepsilon^{(k^\prime)} \\
  = & \frac{1}{\sqrt{N}} \sum_{k=1}^K \frac{1}{n_{k^\prime}} S_{\alpha c}^{(k)}S_\varepsilon^{(k^\prime)} + \frac{\E[\alpha_0(Z_{it})]}{\sqrt{N}} \sum_{k=1}^K \frac{n_k}{n_{k^\prime}} S_\varepsilon^{(k^\prime)} \\
  = & \frac{1}{\sqrt{N}} \sum_{k=1}^K \frac{1}{n_{k^\prime}} S_{\alpha c}^{(k)}S_\varepsilon^{(k^\prime)} + \frac{\E[\alpha_0(Z_{it})]}{\sqrt{N}} \sum_{k=1}^K S_\varepsilon^{(k^\prime)} + o_p(1).
\end{align*}
By the CLT, $\displaystyle \frac{1}{\sqrt{n_k}} S_{\alpha c}^{(k)} = O_p(1)$ and $\displaystyle \frac{1}{\sqrt{n_{k^\prime}}}S_\varepsilon^{(k^\prime)} = O_p(1)$. Thus 
\[
\frac{1}{\sqrt{N}} \sum_{k=1}^K \frac{1}{n_{k^\prime}} S_{\alpha c}^{(k)}S_\varepsilon^{(k^\prime)} = \frac{1}{\sqrt{N}} \sum_{k=1}^K \frac{\sqrt{n_k}}{\sqrt{n_{k^\prime}}}\frac{1}{\sqrt{n_k}} S_{\alpha c}^{(k)} \frac{1}{\sqrt{n_{k^\prime}}}S_\varepsilon^{(k^\prime)} = O_p\Big(\frac{1}{\sqrt{N}}\Big) = o_p(1),
\]
which by the CLT and the Slutsky's theorem implies 
\begin{align*}
  & \sqrt{N} \Big(\widehat{\theta}_{0t}^d(s) - \theta_{0t}(s)\Big) \\
  = & \frac{1}{\sqrt{N}}\sum_{i=1}^N \Big\{\frac{\partial \gamma_0(V_{it})}{\partial D_{i,t-s}} - \theta_{0t}(s) + \alpha_0(Z_{it})\Delta \varepsilon_{it}\Big\} - \frac{\E[\alpha_0(Z_{it})]}{\sqrt{N}} \sum_{k=1}^K S_\varepsilon^{(k)} + o_p(1)\\
  = & \frac{1}{\sqrt{N}}\sum_{i=1}^N \Big(\frac{\partial \gamma_0(V_{it})}{\partial D_{i,t-s}} - \theta_{0t}(s) + \alpha_0(Z_{it})\Delta \varepsilon_{it} - \E[\alpha_0(Z_{it})]\Delta \varepsilon_{it}\Big) + o_p(1) \\
  & \hspace{-0.5cm} \dto N(0,\Psi).
\end{align*}
Note that we can use the CLT because Assumption \ref{assa} implies the variance of $\Big(\frac{\partial \gamma_0(V_{it})}{\partial D_{i,t-s}} - \theta_{0t}(s) + \alpha_0(Z_{it})\Delta \varepsilon_{it} - \E[\alpha_0(Z_{it})]\Delta \varepsilon_{it}\Big)$ is finite. 
\end{proof}
\begin{proof}[Proof of Theorem \ref{prop1}]
Denote $\widehat{\psi}_i = \frac{\partial \widehat{\gamma}_k(V_{it})}{\partial D_{i,t-s}} - \widehat{\theta}_{0t}^d(s) + \big[\widehat{\alpha}_k(Z_{it}) - \frac{1}{n_k}\sum_{j\in\mathcal{F}_k}\widehat{\alpha}_k(Z_{jt})\big][\Delta Y_{it}^* - \Delta \widehat{\gamma}_k^*(V_{it})]$ and $\widetilde{\psi}_i = \frac{\partial \gamma_0(V_{it})}{\partial D_{i,t-s}} - \theta_{0t}(s) + \big[\alpha_0(Z_{it}) - \E[\alpha_0(Z_{it})]\big][\Delta Y_{it}^* - \Delta \gamma_0^*(V_{it})]$.

Note that $\widetilde{\psi}_i = \frac{\partial \gamma_0(V_{it})}{\partial D_{i,t-s}} - \theta_{0t}(s) + \big[\alpha_0(Z_{it}) - \E[\alpha_0(Z_{it})]\big]\big[\Delta \varepsilon_{it} - \frac{1}{n_{k^\prime}}\sum_{u\in\mathcal{F}_{k^\prime}}\Delta \varepsilon_{ut}\big]$ and we denote $\psi_i = \frac{\partial \gamma_0(V_{it})}{\partial D_{i,t-s}} - \theta_{0t}(s) + \big[\alpha_0(Z_{it}) - \E[\alpha_0(Z_{it})]\big]\Delta \varepsilon_{it}$, $B_i = \alpha_0(Z_{it}) - \E[\alpha_0(Z_{it})]$, and $\overline{\varepsilon}_{k^\prime} = \frac{1}{n_{k^\prime}}\sum_{u\in\mathcal{F}_{k^\prime}}\Delta \varepsilon_{ut}$. Then, by the LLN, $\overline{\varepsilon}_{k^\prime} = o_p(1)$, so we have 
\begin{align}
  & \frac{1}{N} \sum_{i=1}^N \widetilde{\psi}_i^2 = \frac{1}{N} \sum_{k=1}^K \sum_{i\in\mathcal{F}_k} \big(\psi_i - B_i\overline{\varepsilon}_{k^\prime}\big)^2 \notag \\
  = & \frac{1}{N} \sum_{i=1}^N \psi_i^2 - \frac{2}{N} \sum_{k=1}^K \sum_{i\in\mathcal{F}_k} \psi_i B_i\overline{\varepsilon}_{k^\prime} + \frac{1}{N} \sum_{k=1}^K \sum_{i\in\mathcal{F}_k} B_i^2\overline{\varepsilon}_{k^\prime}^2 \notag \\
  = & \frac{1}{N} \sum_{i=1}^N \psi_i^2 - \frac{2}{N} \sum_{k=1}^K \sum_{i\in\mathcal{F}_k} o_p(1) + \frac{1}{N} \sum_{k=1}^K \sum_{i\in\mathcal{F}_k} o_p(1) \notag\\
  = & \frac{1}{N} \sum_{i=1}^N \psi_i^2 + o_p(1) = \E[\psi_i^2] + o_p(1). \label{eq31}
\end{align}

Then, we calculate 
\begin{equation}
  \label{eq32}
  \frac{1}{N} \sum_{i=1}^N \widehat{\psi}_i^2 = \frac{1}{N} \sum_{i=1}^N  \big(\widehat{\psi}_i-\widetilde{\psi}_i + \widetilde{\psi}_i)^2 =  \frac{1}{N} \sum_{i=1}^N (\widehat{\psi}_i-\widetilde{\psi}_i)^2 + \frac{2}{N} \sum_{i=1}^N (\widehat{\psi}_i-\widetilde{\psi}_i)\widetilde{\psi}_i + \frac{1}{N} \sum_{i=1}^N  \widetilde{\psi}_i^2.
\end{equation}
Note that 
\begin{equation}
  \label{eq21}
  (\widehat{\psi}_i-\widetilde{\psi}_i)^2  = (\widehat{R}_{1ki} + \widehat{R}_{2ki} + \widehat{R}_{3ki} + \widehat{R}_{4ki})^2 \leq C(\widehat{R}_{1ki}^2 + \widehat{R}_{2ki}^2 + \widehat{R}_{3ki}^2 + \widehat{R}_{4ki}^2),
\end{equation}
where 
\begin{gather*}
  \widehat{R}_{1ki} = \frac{\partial \widehat{\gamma}_k(V_{it})}{\partial D_{i,t-s}} - \frac{\partial \gamma_0(V_{it})}{\partial D_{i,t-s}}, \\
  \widehat{R}_{2ki} = \Big[\widehat{\alpha}_k(Z_{it}) - \frac{1}{n_k}\sum_{i\in\mathcal{F}_k}\widehat{\alpha}_k(Z_{it})\Big][\Delta \gamma_0(V_{it})^* - \Delta \widehat{\gamma}_k^*(V_{it})], \\
  \widehat{R}_{3ki} = \Big[\widehat{\alpha}_k(Z_{it}) - \frac{1}{n_k}\sum_{i\in\mathcal{F}_k}\widehat{\alpha}_k(Z_{it})- \big(\alpha_0(Z_{it}) - \E[\alpha_0(Z_{it})]\big)\Big][\Delta Y_{it}^* - \Delta \gamma_0^*(V_{it})], \\
  \widehat{R}_{4ki} = \theta_{0t}(s) - \widehat{\theta}_{0t}^d(s).
\end{gather*}

By Assumption \ref{assm}, $\E(\widehat{R}_{1ki}^2 \mid \wc) \pto 0$. Then, 
\begin{align*}
  & \E(\widehat{R}_{2ki}^2 \mid \wc) = \int \Big(\Big[b(z_{it})-\frac{1}{n_k}\sum_{i\in\mathcal{F}_k}b(z_{it})\Big]^\prime \widehat{\rho}_k\Big)^2[\Delta \widehat{\gamma}_k^*(v_{it}) - \Delta \gamma_0^*(v_{it})]^2 F_0(\rd w_{it}^*) \\
  \leq & 4C_b^2\|\widehat{\rho}_k\|_1^2 \int [\Delta \widehat{\gamma}_k^*(v_{it}) - \Delta \gamma_0^*(v_{it})]^2 F_0(\rd w_{it}^*) = O_p(1) \int [\Delta \widehat{\gamma}_k^*(v_{it}) - \Delta \gamma_0^*(v_{it})]^2 F_0(\rd w_{it}^*) \\
  = & O_p(1)o_p(1) = o_p(1).
\end{align*}
where the inequality follows by the almost surely bounded condition for $b(Z_{it})$ in Assumption \ref{assbd}, the second equality follows by Lemma \ref{lemma13}, and the third equality is implied by Assumption \ref{assm}. Next, note that 
\begin{align*}
  & \Big[\widehat{\alpha}_k(Z_{it}) - \frac{1}{n_k}\sum_{i\in\mathcal{F}_k}\widehat{\alpha}_k(Z_{it})- \big(\alpha_0(Z_{it}) - \E[\alpha_0(Z_{it})]\big)\Big]^2 \\
  \leq & C\Big\{[\widehat{\alpha}_k(Z_{it}) - \alpha_0(Z_{it})]^2 + \Big[\frac{1}{n_k}\sum_{i\in\mathcal{F}_k}\big[\widehat{\alpha}_k(Z_{it}) - \alpha_0(Z_{it})\big]\Big]^2 \\
  & \hspace{6cm} + \Big[\frac{1}{n_k}\sum_{i\in\mathcal{F}_k}\alpha_0(Z_{it}) - \E[\alpha_0(Z_{it})]\Big]^2\Big\}.
\end{align*}
By Assumption \ref{assm}, 
\[
\E\Big([\widehat{\alpha}_k(Z_{it}) - \alpha_0(Z_{it})]^2 \mid \wc\Big) = \int [\widehat{\alpha}_k(z_{it}) - \alpha_0(z_{it})]^2 F_0(\rd z_{it}) \pto 0.
\]
By the Cauchy-Schwarz inequality and Assumption \ref{assm}, 
\[
\E\Big(\Big[\frac{1}{n_k}\sum_{i\in\mathcal{F}_k}\big[\widehat{\alpha}_k(Z_{it}) - \alpha_0(Z_{it})\big]\Big]^2 \mid \wc\Big) \leq \E\Big([\widehat{\alpha}_k(Z_{it}) - \alpha_0(Z_{it})]^2 \mid \wc\Big) \pto 0.
\]
In addition,
\begin{align*}
  & \E\Big(\Big[\frac{1}{n_k}\sum_{i\in\mathcal{F}_k}\alpha_0(Z_{it}) - \E[\alpha_0(Z_{it})]\Big]^2 \mid \wc\Big) = \E\Big[\frac{1}{n_k}\sum_{i\in\mathcal{F}_k}\alpha_0(Z_{it}) - \E[\alpha_0(Z_{it})]\Big]^2 \\
  & = \frac{1}{n_k} \Var(\alpha_0(Z_{it})) \to 0,
\end{align*}
where the variance is finite because $\alpha_0(Z_{it}) \leq C_\alpha$ in Assumption \ref{assa}. Thus, we obtain $\E(\widehat{R}_{3ki}^2 \mid \wc) \pto 0$. 

Combining the above three results, we derive $\E(\widehat{R}_{1ki}^2 + \widehat{R}_{2ki}^2 + \widehat{R}_{3ki}^2 \mid \wc) \pto 0$, which, by the conditional Markov inequality, implies 
\begin{equation}
  \label{eq22}
  \frac{1}{N} \sum_{i=1}^N \Big(\widehat{R}_{1ki}^2 + \widehat{R}_{2ki}^2 + \widehat{R}_{3ki}^2 \Big) \pto 0.
\end{equation}
By Theorem \ref{th2}, $\widehat{\theta}_{0t}^d(s) - \theta_{0t}(s) = o_p(1)$, implying 
\begin{equation}
  \label{eq23}
  \frac{1}{N}\sum_{i=1}^N \widehat{R}_{4ki}^2 \pto 0.
\end{equation}
Therefore, it follows from \eqref{eq21}, \eqref{eq22}, and \eqref{eq23} that 
\begin{equation}
  \label{eq33}
  \frac{1}{N}\sum_{i=1}^N (\widehat{\psi}_i-\widetilde{\psi}_i)^2 = o_p(1).
\end{equation}

By the Cauchy-Schwarz inequality, 
\begin{equation}
  \label{eq34}
  \frac{2}{N} \sum_{i=1}^N (\widehat{\psi}_i-\widetilde{\psi}_i)\widetilde{\psi}_i \leq 2 \sqrt{\frac{1}{N} \sum_{i=1}^N (\widehat{\psi}_i-\widetilde{\psi}_i)^2}\sqrt{\frac{1}{N} \sum_{i=1}^N \widetilde{\psi}_i^2} = 2\sqrt{o_p(1)}\sqrt{\E[\psi_i^2] + o_p(1)} = o_p(1).
\end{equation}

Finally, it follows from \eqref{eq31}, \eqref{eq32}, \eqref{eq33}, and \eqref{eq34} that 
\[
\frac{1}{N} \sum_{i=1}^N \widehat{\psi}_i^2 = \E[\psi_i^2] + o_p(1) = \Psi + o_p(1),
\]
which completes the proof.
\end{proof}

\subsection{Proof of theorems in the appendix}
\begin{proof}[Proof of Theorem \ref{ths1}]
From the proof of Theorem \ref{th2}, we have 
\[
\begin{aligned}
  & \sqrt{N} \Big(\widehat{\theta}_{0t}^d(s) - \theta_{0t}(s)\Big) \\
= & \frac{1}{\sqrt{N}}\sum_{i=1}^N \Big(\frac{\partial \gamma_0(V_{it})}{\partial D_{i,t-s}} - \theta_{0t}(s) + \alpha_{0ts}(Z_{it})\Delta \varepsilon_{it} - \E[\alpha_{0ts}(Z_{it})]\Delta \varepsilon_{it}\Big) + o_p(1).
\end{aligned}
\]
Combining the results for $s=0,\ldots,q$, by the multivariate CLT and the Slutsky's theorem, we have  
\[
\begin{aligned}
  & \sqrt{N} \begin{pmatrix}
    \widehat{\theta}_{0t}^d(0) - \theta_{0t}(0) \\
    \cdots \\
    \widehat{\theta}_{0t}^d(q) - \theta_{0t}(q)
  \end{pmatrix} \\
= & \frac{1}{\sqrt{N}}\sum_{i=1}^N \begin{pmatrix}
  \frac{\partial \gamma_0(V_{it})}{\partial D_{i,t}} - \theta_{0t}(0) + \alpha_{0t0}(Z_{it})\Delta \varepsilon_{it} - \E[\alpha_{0t0}(Z_{it})]\Delta \varepsilon_{it} \\
  \cdots \\
  \frac{\partial \gamma_0(V_{it})}{\partial D_{i,t-q}} - \theta_{0t}(q) + \alpha_{0tq}(Z_{it})\Delta \varepsilon_{it} - \E[\alpha_{0tq}(Z_{it})]\Delta \varepsilon_{it} 
\end{pmatrix} + o_p(1) \\
& \hspace{-0.5cm} \dto N(0, \Sigma).
\end{aligned}
\]
Therefore, 
\[
\sqrt{N} \Big(\widehat{\theta}_{0t} - \theta_{0t}\Big) = \sqrt{N} w^\prime \begin{pmatrix}
    \widehat{\theta}_{0t}^d(0) - \theta_{0t}(0) \\
    \cdots \\
    \widehat{\theta}_{0t}^d(q) - \theta_{0t}(q)
  \end{pmatrix} \dto N(0,w^\prime \Sigma w).
\]
\end{proof}

\begin{proof}[Proof of Theorem \ref{ths2}]
  The proof is similar to the proof of Theorem \ref{ths1} and is thus omitted.
\end{proof}
\end{appendices}

\bibliography{bibliography.bib}

\end{document}